%% file: parity-phase-reduction-ZX-TODD.tex
\documentclass[submission,copyright,creativecommons]{eptcs}
\providecommand{\event}{QPL 2019}
\usepackage{breakurl}             
\usepackage{underscore}           
\usepackage{fixltx2e}
\usepackage{tikz}
\usetikzlibrary{shapes,arrows,calc,positioning,fit}
\usepackage{pgfplots}
\usetikzlibrary{trees}
\usetikzlibrary{topaths}
\usetikzlibrary{decorations.pathmorphing}
\usetikzlibrary{decorations.markings}
\usetikzlibrary{matrix,backgrounds,folding}
\usetikzlibrary{chains,scopes,positioning,fit}
\usetikzlibrary{calc}
\usetikzlibrary{chains}
\usetikzlibrary{shapes,shapes.geometric,shapes.misc}\pgfdeclarelayer{edgelayer}
\pgfdeclarelayer{nodelayer}
\pgfsetlayers{background,edgelayer,nodelayer,main}
\tikzstyle{none}=[inner sep=0mm]
\tikzstyle{every loop}=[]
\input{tikzstyles.tex}

\input{defs.tex}

\usepackage{scalefnt}
\usepackage{amsmath,amssymb,amsthm,bbm,mathtools,stmaryrd}
\usepackage[utf8x]{inputenc}
\usepackage{enumitem}
\usepackage[font={small}]{caption}
\usepackage[mathscr]{euscript}

\newcommand\ie{\emph{i.e.}}
\newcommand\eg{\emph{e.g.}}

\edef\pp{parity-phase}
\edef\Pp{Parity-phase}
\providecommand\pipp{$\pi\!\!\:/4\!\:$-\pp}

\newcommand\parStyle[1]{\textrm{\mdseries\upshape({#1}\kern0.1ex)}}

\newlength\romanumlabelwd
\makeatletter
  \renewenvironment{abstract}{%
      \if@twocolumn
        \subsection*{\abstractname}%
        \small
      \else
        \small
        \begin{center}%
          {\bfseries \abstractname\vspace{-.5em}\vspace{\z@}}%
        \end{center}%
        \quotation
      \fi}
      {\if@twocolumn\par\else\endquotation\fi}

\makeatother

\makeatletter
\def\tagform@#1{%
	\ifmmode
		\mbox{\normalsize\maketag@@@{(\ignorespaces#1\unskip\@@italiccorr)}}%
	\else
		\maketag@@@{(\ignorespaces#1\unskip\@@italiccorr)}%
	\fi}
\makeatother

\newtheoremstyle{theorem?}
  {\topsep}{\topsep}   																	
  {\itshape}{0pt}{\bfseries}{?}{5pt plus 1pt minus 1pt} 
  {}          																					
\theoremstyle{theorem?}

\makeatletter
\def\@cludgescbf#1#2\end{\textbf{#1\scalefont{0.85}#2}}
\newcommand\cludgescbf[1]{\@cludgescbf#1\end}
\makeatother

\newtheoremstyle{Axiom}
  {\topsep}{\topsep}   																	
  {\itshape}{0pt}{\bfseries}{?}{5pt plus 1pt minus 1pt} 
  {}          																					
\theoremstyle{Axiom}

\theoremstyle{definition}

\newtheorem*{definition*}{Definition}

\theoremstyle{plain}
\newtheorem{theorem}{Theorem}
\newtheorem{lemma}[theorem]{Lemma}

\newtheorem*{proposition*}{Proposition}

\newtheorem*{claim*}{Claim}

\newtheorem*{tactic*}{{\large\scshape phage tactic}}

\theoremstyle{definition}

\newlist{algenum}{enumerate}{9}
\setlist[algenum,1]{%
  label=\arabic*.,
  itemsep=0ex,
  topsep=1ex,
  leftmargin=2em}
\setlist[algenum,2]{%
  label=\alph*\upshape.,
  itemsep=0ex,
  topsep=-0.375ex,
  leftmargin=1.5em}
\setlist[algenum,3]{%
  label=(\!\!\;\itshape\roman*\;\!\upshape),
  topsep=0ex,
  itemsep=0.375ex,
  leftmargin=2.5em}

\newcommand\C{\mathbb{C}}

\newcommand\Z{\mathbb{Z}}

\newcommand\e{\mathrm{e}}

\renewcommand\vec\mathbf

\newcommand\herm{^\dagger}

\newcommand\ox{\otimes}
\newcommand\x{\times}
\newcommand\sox[1]{^{\otimes #1}}
\newcommand\sur[1]{^{(#1)}}


\let\prsubset\subset
\let\subset\subseteq
\let\oldepsilon\epsilon
\let\epsilon\varepsilon
\let\varepsilon\oldepsilon
\let\le\leqslant
\let\ge\geqslant

\newcommand\idop{\mathbbm{1\!}}

\providecommand\href[2]{\texttt{#2}}

\makeatletter
	\newcommand\ket[1]{\left| #1 \right\rangle\@ifnextchar\bra{\mspace{-4mu}}{}}
	\newcommand\bra[1]{\left\langle #1 \right|}

\makeatother

\makeatletter
\newif\if@ref
\@reffalse
\renewcommand\subsubsection{\@startsection{subsubsection}{3}{\z@}%
                                     {-3ex\@plus -1ex \@minus -.2ex}%
                                     {1.5ex \@plus .2ex}%
                                     {\let\thesubsection\relax\centering\normalfont\small\itshape\bfseries}}

\makeatother

\usepackage{hyperref}

\def\thetitle{
  ~Techniques to Reduce $\pi/4$-Parity-Phase Circuits,~ \\
  Motivated by the ZX Calculus%
}
\title{\thetitle}
\def\titlerunning{\thetitle}
\def\authorrunning{%
  N.\;de\;Beaudrap, X.\;Bian, \& Q.\;Wang
}%
\author{%
  Niel de Beaudrap
  \institute{%
    \hspace*{-0.5em} Department of Computer Science \hspace*{-0.5em}\\
    University of Oxford \\
    Oxford, UK
  }
  \email{\hspace*{-0.5em} niel.debeaudrap@cs.ox.ac.uk \hspace*{-0.5em}}
  \and
  Xiaoning Bian
  \institute{%
    \hspace*{-0.5em} Department of Mathematics \& Statistics \hspace*{-0.5em} \\
    Dalhousie University \\
    Halifax, Canada
  }
  \email{bian@dal.ca }
  \and
  Quanlong Wang%
  \institute{%
    \hspace*{-0.5em} Department of Computer Science \hspace*{-0.5em} \\
    University of Oxford \\
    Oxford, UK
  }
  \institute{%
    \hspace*{-2.5ex} Cambridge Quantum Computing Ltd. \hspace*{-2.5ex} \\
    Cambridge, UK
  }
  \email{\hspace*{-0.5em} quanlong.wang@cs.ox.ac.uk \hspace*{-0.5em}}
}

\begin{document}
\maketitle

\newcommand{\tikzfig}[1]{%
          \input{./#1.tikz}
}
\input{zx.tikzstyles}

\vspace*{-5mm}
\begin{abstract}
  To approximate arbitrary unitary transformations on one or more qubits, one must perform transformations which are outside of the Clifford group.
  The gate most commonly considered for this purpose is the $T = \mathrm{diag}(1,\e^{i\pi\!\!\:/4})$ gate.
  As $T$ gates are computationally expensive to perform fault-tolerantly in the most promising error-correction technologies, minimising the ``$T$-count'' (the number of $T$ gates) required to realise a given unitary in a Clifford+$T$ circuit is of great interest.
  We describe techniques to find circuits with reduced $T$-count in unitary circuits, which develop on the ideas of Heyfron and Campbell~\cite{HC-2018} with the help of the ZX calculus.
  Following Ref.~\cite{HC-2018}, we reduce the problem to that of minimising the $T$ count of a CNOT+$T$ circuit.
  The ZX calculus motivates a further reduction to simplifying a product of commuting ``\pipp'' operations: diagonal unitary transformations which induce a relative phase of $\e^{i\pi/4}$ depending on the outcome of a parity computation on the standard basis (which motivated Kissinger and van~de~Wetering~\cite{KvdW-2019} to introduce ``phase gadgets'').
  For a number of standard benchmark circuits, we show that these techniques --- in some cases supplemented by the TODD subroutine of Heyfron and Campbell~\cite{HC-2018} --- yield $T$-counts comparable to or better than the best previously known results.
\end{abstract}


\vspace*{-1ex}
\section{Introduction}

An important goal of quantum technologies is to realise, as faithfully as possible, an architecture capable of performing approximately universal quantum computation.
Ignoring the practical difficulties of imperfectly realised operations and noise in imperfect hardware, such an architecture must be able to approximate an arbitrary unitary transformation with high probability, possibly relative to some embedding of the computational space into the states of its qubits (\eg,~to perform error correction).

The above goal requires that the set of transformations that the architecture can perform do not form a discrete set.
This is challenging, as the operations which can be easily performed fault-tolerantly for various error correcting codes form a discrete set --- often the Clifford group, or a subset of it.
As the Clifford group is in any case useful to reason about quantum error correction and very simple procedures on quantum data, this motivates  \textrm{\bfseries (a)}~considering fault-tolerant realisations of the Clifford group, together with a more labour-intensive procedure to realise some unitary transformation outside of the Clifford group, and then \textrm{\bfseries (b)}~minimising the number of non-Clifford gates required to realise or approximate a given unitary.
The most popular approach is to consider ``Clifford+T'' circuits, using a gate-set such as $\{\mathrm{CNOT}, H, S, T\}$, involving CNOT, the Hadamard gate $H$, and $S = \mathrm{diag}(1,i)$ as generators of the Clifford group,
supplemented by the gate $T = \mathrm{diag}(1, \e^{i\pi/4}) = \sqrt{S}$.
We then consider the problem of minimising the \emph{T-count} of a unitary transformation: the number of $T$ gates to realise (or approximate) that unitary.

Heyfron and Campbell~\cite{HC-2018} describe a circuit transformation that allows one to realise a Clifford+T unitary using a circuit consisting of a circuit of CNOT operations,
a circuit of
\iffalse
$\mathcal D_3$
\else
  diagonal non-Clifford
\fi
operations, and a sequence of (possibly classically controlled) Clifford operations.
This allows them to reduce the problem of $T$-count reduction to an appropriate analysis of the
\iffalse
$\mathcal D_3$
\else
diagonal non-Clifford
\fi
portion of this circuit.
The strategy of Heyfron and Campbell~\cite{HC-2018} is to consider non-Clifford diagonal circuit in terms of \emph{phase polynomials}, and builds on a sequence of results which revolve around such operations~\cite{AMMR-2013,GKMR-2014,AMM-2014,ACR-2018,CH-2017,AM-2019} presented in various but similar ways.
These results note the connection of $T$-count optimisation to difficult coding problems and NP-hard tensor decomposition problems~\cite{AM-2019,HC-2018}, and generally approach the problem of reducing $T$-count by approaching these difficult problems.

Our approach is to describe diagonal CNOT+$T$ unitaries using ``\pipp'' operations.
These are operations which induce a $\e^{i\pi/4}$ phase on standard basis states, depending on a parity computation $f(x) = x_{k_1} \oplus x_{k_2} \oplus \cdots \oplus x_{k_m}$\,, for any integer $m \ge 1$, and $1 \le k_1, k_2, \ldots, k_m \le n$.
As each \pipp\ gate can be realised in principle using a single $T$ or $T\herm$ gate (and some CNOT gates), simplifying \pipp\ circuits is directly productive to reducing $T$-count.
On this same line of investigation, Kissinger and van~de~Wetering~\cite{KvdW-2019} use the ZX calculus to describe a technique of ``phase teleportation'' to reduce circuits involving ``phase gadgets'' (denoting unitaries such as our \pipp\ operations).

In this article, we describe a framework to reduce $T$-count, by using ``tactics'' which are induced by any family of identities on \pipp\ operations.
We then describe some identities on \pipp\ operations (which define two different such tactics), and describe strategies to deploy these tactics in an effective way.
Our techniques yield new records for the $T$-count of some standard benchmark circuits, and yield results which are near to the best known results in further circuits.
Because of the simple way in which we use these identities on \pipp\ operations, we speculate that even these record-setting results may be easy to improve on.

\vspace*{-1ex}
\section{Preliminaries}

We first set out some basic or existing results, using the following notation.
Let $[n] := \{1,2,\ldots,n\}$ and $\idop$ be the $2 \!\x\! 2$ identity matrix.
For sets $S, T \subset V$ we write $S \mathbin\Delta T$ for the symmetric difference $(S \cup T) \setminus (S \cap T)$, and $\vec x\sur{S} \in \{0,1\}^V$ denote the incidence vector of $S$, where $x^{\,(S)}_{\!\!\!\:\smash{j}} = 1$ if and only if $j \in S$.

\vspace*{-1ex}
\subsection{The Clifford hierarchy} 

Let $\mathcal P^n := \bigl\{ i^k \!\!\; P_1 \ox \cdots \ox P_n \!\mathrel{\big\vert}\! k \!\in\! \Z \mathrel\& P_j \!\in\! \{\idop,\!\!\;X,\!\!\;Y,\!\!\;Z\} \!\bigr\}$ denote the $n$-qubit Pauli group.
We define the Clifford hierarchy (on $n$ qubits) by defining $\mathcal C_1^n := \mathcal P_n$, and
\vspace*{-1ex}
\begin{align}
\label{eqn:CliffordHierarchy}%
  \mathcal C_k^n = \bigl\{ U \!\in\! \mathrm U_n(\C) \mathbin{\bigl\lvert}
                    \forall P \!\in\! \mathcal P^n\!.\; U\!\!\: PU\herm \!\in\!\!\; \mathcal C_{k\text{--}1}^n \bigr\}
\end{align}
\vspace*{-3.5ex}

\noindent for $k > 1$. 
We then define $\mathcal D_k^n \subseteq \mathcal C_k^n$ to be the subset of diagonal operations.
As an abuse of notation, we will identify $\mathcal C_k^n$ and $\mathcal D_k^n$ with subsets of $\mathcal C_k^N$ and $\mathcal D_k^N$ (respectively) for $n < N$. 
As a part of this abuse of notation, we allow ourselves to write $S \in \mathcal C_2^n$ and $T \in \mathcal C_3^n$ for all $n \ge 1$. 



\vspace*{-1ex}
\subsection{\Pp\ operations}

\paragraph{Defining \pp\ operations.}
It is easy to show that $\mathcal D_k^n$ forms an abelian group.
In particular, one can show (see Appendix~\ref{apx:generators-Dkn}) that $\mathcal D_k^n$ is generated by the operators $\omega \cdot \idop\sox{n}$ for any global phase $\omega$, together with all operations of the form $D_{\!\!\:S,k}$ for sets $S = \{s_1,\ldots,s_m\} \subseteq [n]$ for $m \ge 1$, defined by
\vspace*{-.5ex}
\begin{equation}
\begin{aligned}[b]
    D_{\!\!\:S,k}
  \;=\;
    \exp\Bigl(-\frac{\text{\small$i\pi$}}{2^{k}}\bigl(Z_{s_1} \!\ox \cdots \ox\! Z_{s_m}\bigr)\Bigr)
  \;=\;
    \exp\Bigl(-\frac{\text{\small$i\pi$}}{2^{k}} \!\; Z_S\Bigr)
  \;=\;
    \cos\bigl(\tfrac{\pi}{2^{k}}\bigr) \idop - i \sin\bigl(\tfrac{\pi}{2^{k}}\bigr) Z_S\,,
\end{aligned}
\end{equation}
\vspace*{-2.5ex}

\noindent
where $Z_S = \smash{\bigotimes_{j\in S} Z_j}$\,.
(We define $D_{\!\!\:S,k}$ for all $k \in \Z$; however, one may show  $D_{\!\!\:S,0} \!=\! - \idop\sox{n}$ and $D_{\!\!\:S,\!\!\;-k} \!=\! \idop\sox{n}$ for all $k>0$ and $S \subseteq [n]$.)
Note that $\smash{X_j Z_S X_j\herm} = \smash{(-1)^{x\sur{S}_{\!\!\;j}}\!\: Z_S}$, and  $\mathrm{CNOT}_{h,j}\,Z_S\,\mathrm{CNOT}_{h,j}\herm = Z_{S'}$ such that
\vspace*{-1.5ex}
\begin{equation}
  S'  \;=\;
  \begin{cases}
    S \mathbin\Delta\{h\}, & \text{if $j \in S$};
  \\
    S,  & \text{otherwise}.      
  \end{cases}
\end{equation}%
\vspace*{-2.5ex}

\begin{subequations}%
\label{eqn:conjugate-Dk-NOTs}\noindent
From this it follows that
\vspace*{-2.5ex}
\begin{align}
  \label{eqn:conjugate-Dk-X}
    X_j \,D_{\!\!\:S,k} \,X_j\herm
  \;&=\;
    D_{\!\!\:S,k}^{-1}
  \;\in\;
    \mathcal D_k^n
  \\[-5.5ex]\notag
\intertext{%
  if $j \in S$ (and $X_j \,D_{\!\!\:S,k} \,X_j\herm = D_{\!\!\:S,k}$ otherwise); and
}
  \notag\\[-4.5ex]
  \label{eqn:conjugate-Dk-CNOT}
    \mathrm{CNOT}_{h,j}\,D_{\!\!\:S,k}\,\mathrm{CNOT}_{h,j}\herm
  \;&=\;
    D_{\!\!\:S'\!\!\!\:,\:\!k}
  \;\in\; \mathcal D_k^n
\end{align}%
\end{subequations}%
\vspace*{-3ex}

\noindent
so that $\mathcal D_k^n$ is preserved under conjugation by CNOT and $X$ operations.
Also note that $D_{\!\!\:S,\:\!k}^{\;2} = D_{\!\!\:S,\:\!k{-}1}$, from which it follows that $\mathcal D_{k{-}1}^n \subset \mathcal D_k^n$.

We refer to operations $D_{\!\!\:S,k{+}1}$\,, and their inverses, as ``$\pi\!\!\:/2^k\!\:$-\pp'' operations.
We motivate this terminology as follows.
Let $S = \{s_1, s_2, \ldots, s_m\}$ for some $m \ge 1$.
For a standard basis vector $\ket{z}$, we have $Z_S \ket{z} = (-1)^{\vec x\sur{S} \cdot\;\! z} \ket{z}$, where we define $\vec x\sur{S} \cdot z = \sum_i x^{\,\smash{(S)}}_i z_i$\,.
From this it follows that
\vspace*{-1ex}
\begin{equation}
  D_{\!\!\:S,k{+}1} \ket{z} =
  \begin{cases}
      \exp(-i\pi/2^{k{+}\!1}) \ket{z} , 
          & \text{if $\vec x\sur{S} \cdot z = 0$};
    \\[0.5ex]
      \exp(+i\pi/2^{k{+}\!1}) \ket{z} , 
          & \text{if $\vec x\sur{S} \cdot z = 1$}.
  \end{cases}
\end{equation}

\vspace*{-1ex}
\noindent
This is equivalent (up to a global phase of $\e^{-i\pi\!\!\:/2^{k{+}\!1}}$) to inducing a relative phase of $\pi\!\!\:/2^k$ on $\ket{z}$ for those $z \in \{0,1\}^n$ for which $\vec x\sur{S} \cdot z = z_{s_1} \oplus z_{s_2} \oplus \cdots \oplus z_{s_m} = 1$; and similarly for $D_{\!\!\:S,{k{+}1}}^{-1}$.
More generally, we refer to $\exp(\pm \tfrac{1}{2}i \theta Z_S)$ as a $\theta\!\!\;$-\pp\ operation.
We note that $\theta\!\!\;$-phase parity operation, the operators $D_{\!\!\:S,k}$ among them, can be represented by ZX diagrams with the usual denotational semantics (read from left to right in this article), with structure such as the following:
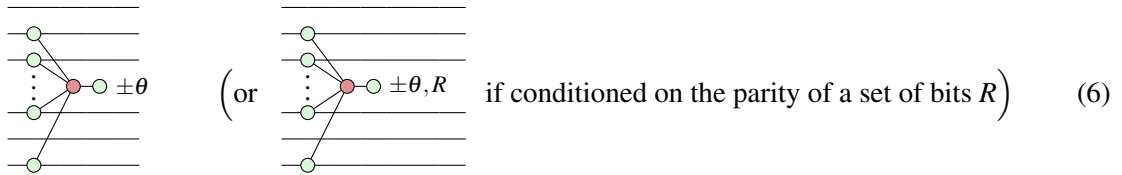
\vspace*{-.5ex}
\begin{equation}
  \label{eqn:diagramD-Sk}
    \begin{aligned}
    \begin{tikzpicture}[scale=0.875]
      \def\dx{0.4}
      \def\dy{0.4}
      \coordinate (0) at (0,0);
      \coordinate (x0-0) at (0);
      \coordinate (x1-0) at ($(x0-0) + (0,\dy)$);
      \coordinate (x2-0) at ($(x1-0) + (0,\dy)$);
      \coordinate (x3-0) at ($(x2-0) + (0,\dy)$);
      \coordinate (x4-0) at ($(x3-0) + (0,\dy)$);
      \coordinate (x5-0) at ($(x4-0) + (0,\dy)$);
      \coordinate (x6-0) at ($(x5-0) + (0,\dy)$);
      \xdef\u{0}
      \foreach \t in {1,...,5} {%
        \foreach \k in {0,...,6} {%
          \coordinate (x\k-\t) at ($(x\k-\u) + (\dx,0)$);
          \ifnum\k=3%
            \ifnum\t=2
              \coordinate (x3-2) at ($(x3-2) + ({\dx/2},0)$);
            \fi
          \else
            \draw (x\k-\u) -- (x\k-\t);
          \fi
        }
        \xdef\u{\t}
      }
      \node at ($(x3-1) + (0,0.1)$) {$\vdots$};
      \foreach \k in {0,2,4,5} {%
        \draw (x\k-1) -- (x3-2);
        \filldraw [style=Z dot] (x\k-1) circle (3pt);
      }
      \draw (x3-2) -- (x3-3);
      \filldraw [style=X dot] (x3-2) circle (3pt);
      \filldraw [style=Z dot] (x3-3) circle (3pt) node [anchor=180,font=\footnotesize] {\;$\pm \theta$}; 
    \end{tikzpicture}
    \end{aligned}
    \qquad
    \Bigl(\text{or }
    \;\;
    \begin{aligned}
    \begin{tikzpicture}[scale=0.875]
      \def\dx{0.4}
      \def\dy{0.4}
      \coordinate (0) at (0,0);
      \coordinate (x0-0) at (0);
      \coordinate (x1-0) at ($(x0-0) + (0,\dy)$);
      \coordinate (x2-0) at ($(x1-0) + (0,\dy)$);
      \coordinate (x3-0) at ($(x2-0) + (0,\dy)$);
      \coordinate (x4-0) at ($(x3-0) + (0,\dy)$);
      \coordinate (x5-0) at ($(x4-0) + (0,\dy)$);
      \coordinate (x6-0) at ($(x5-0) + (0,\dy)$);
      \xdef\u{0}
      \foreach \t in {1,...,7} {%
        \foreach \k in {0,...,6} {%
          \coordinate (x\k-\t) at ($(x\k-\u) + (\dx,0)$);
          \ifnum\k=3%
            \ifnum\t=2
              \coordinate (x3-2) at ($(x3-2) + ({\dx/2},0)$);
            \fi
          \else
            \draw (x\k-\u) -- (x\k-\t);
          \fi
        }
        \xdef\u{\t}
      }
      \node at ($(x3-1) + (0,0.1)$) {$\vdots$};
      \foreach \k in {0,2,4,5} {%
        \draw (x\k-1) -- (x3-2);
        \filldraw [style=Z dot] (x\k-1) circle (3pt);
      }
      \draw (x3-2) -- (x3-3);
      \filldraw [style=X dot] (x3-2) circle (3pt);
      \filldraw [style=Z dot] (x3-3) circle (3pt) node [anchor=180,font=\footnotesize] {\;$\pm \theta, R$}; 
    \end{tikzpicture}
    \end{aligned}\;\;\;
    \text{if conditioned on the parity of a set of bits $R$}\Bigr)
\end{equation}
\vspace*{-2ex}

\noindent
where the long horizontal wires are the qubits indexed by $[n] = \{1,2,\ldots,n\}$ and $S \subseteq [n]$ is the subset of those qubits which have (light, green) degree-3 nodes on them.
In the right-hand diagram, $R$ denotes a set of boolean variables $s_i \in \{0,1\}$: using the extended annotations of Ref.~\cite{DP-2010}, the diagram  denotes that the phase applied is $\pm \theta$ only if $\bigoplus_{s_i \in R} s_i = 1$, and that otherwise the phase is zero.
We refer to these as ``phase gadgets'', adopting the terminology of Ref.~\cite[Section~4.3]{KvdW-2019}; when $\lvert S \rvert = m$, we may refer to it as a ``phase $m$-gadget''.
(If $\theta$ is an odd multiple of $\pi/4$, we may refer to it as a ``$T$-phase $m$-gadget''; for $\theta$ an integer multiple of $\pi/2$, we refer to it as a ``Clifford-phase $m$-gadget''.
If $m=1$, we may also mildly abuse this terminology to refer to  a simple green phase node as a ``$1$-gadget''.)

\vspace*{-2ex}
\paragraph{\Pp\ operations in relation to controlled phases.}
An important role of $D_{S,3}$ gates for $S \subseteq [n]$ is their relationship to  diagonal gates in $\mathcal D_3^n$ which are controlled-unitaries of a more straightforward sense, such as $\mathrm CS$ and $\mathrm {CC}Z$:
\vspace*{-1ex}
\begin{align}
  \begin{split}
    \mathrm CS
    &=  
      \ket{0}\!\!\bra{0} \ox \idop + \ket{1}\!\!\bra{1} \ox S
    \\&=
      \exp\Bigl(\frac{i\pi}{2} \ket{11}\!\!\bra{11}\Bigr),
  \end{split}
&\quad
  \begin{split}
    \mathrm{CC}Z
    &=  
      \Bigl(
        \ket{00}\!\!\bra{00} + \ket{01}\!\!\bra{01} + \ket{10}\!\!\bra{10}
      \Bigr) \ox \idop
      + \ket{11}\!\bra{11} \ox Z
    \\&=
      \exp\Bigl(i \pi \ket{111}\!\!\bra{111}\Bigr);
  \end{split}
\end{align}
\vspace*{-2ex}

\noindent
we may describe how to generate these from $D_{k,3}$ operations by decomposing the projectors $\ket{11}\!\!\bra{11}$ or $\ket{111}\!\!\bra{111}$ into tensor products of $\ket{1}\!\!\bra{1} = \tfrac{1}{2}\bigl(\idop - Z)$, and expanding to obtain a product of $D_{S,3}$ gates (see Eqn.~\eqref{eqn:decomposeCtrlPhase} in Appendix~\ref{apx:generators-Dkn}): disregarding the $D_{\emptyset,3}$ factors, which realise global phases, we obtain
\vspace*{-0.5ex}
\begin{align}
\label{eqn:decomposeCSandCCZ}
    \mathrm CS_{h,j}
    &\propto
      D_{\{h\},3} D_{\{j\},3} D_{\{h,j\},3}^{\,-1}\;;
&\quad
    \mathrm{CC}Z_{g,h,j}
    &\propto
      D_{\{g\},3} D_{\{h\},3} D_{\{j\},3} D_{\{g,h\},3}^{\,-1} D_{\{g,j\},3}^{\,-1} D_{\{h,j\},3}^{\,-1} D_{\{g,h,j\},3}\;.
\end{align}
\vspace*{-2ex}

\noindent
More generally, we may relate ${(t{\:\!-\!\!\;}1\!\!\:)}$-controlled $\pi\!\!\:/2^k$-phase gates to $\pi\!\!\:/2^{k-t+1}$-phase parity gates, following Eqn.~\eqref{eqn:decomposeCtrlPhase}:
\vspace*{-.5ex}
\begin{equation}
    \prod_{\substack{S \in \wp(V) \\ S \ne \emptyset}} D_{S,k}^{(-1)^{\scriptstyle \lvert S \rvert}}
  \propto
    \;\exp\Bigl(\frac{i\pi}{2^{k-\lvert V\rvert+1}}  \ket{1}\!\!\bra{1}\sox{V}\Bigr),
\end{equation}

\vspace*{-.5ex}
\noindent
where the right-hand operator applies a phase of $\pi\!\!\:/2^{k-\lvert T \rvert-1}$ to those components of a state in which all of the qubits in $T$ are in the state $\ket{1}$.
A corollary of this, on which our results depend, is that
\vspace*{-.5ex}
\begin{equation}
  \label{eqn:phaseParityGateIdentity}
    \prod_{\substack{S \in \wp(V) \\ S \ne \emptyset}} D_{S,k}^{(-1)^{\scriptstyle \lvert S \rvert}}
  \propto\;
    \idop\sox{V},
  \quad
  \text{so that}
  \quad
    D_{V\!\!\:,\:\!k}
    \;\propto\!
    \prod_{\substack{S \in \wp(V) \\ S \notin \{\emptyset,V\}}} D_{S,k}^{(-1)^{\scriptstyle \lvert V \rvert - \lvert S \rvert + 1}} \!,
  \qquad\text{for $\lvert V \rvert > k$}.
\end{equation}
\vspace*{-2ex}

\vspace*{-2ex}
\paragraph{Connection between \pp\ operations and $T$-count.}
From Eqn.~\eqref{eqn:conjugate-Dk-CNOT}, it follows that any operation $D_{\!\!\:S,k}$ can be reduced to an operation $D_{\!\!\:{j},k} \propto \mathrm{diag}(1,\e^{2 \pi i/2^k})$ acting on a single qubit $j$, by conjugation with an appropriate CNOT circuit.
In particular, it follows that any $D_{\!\!\:S,3}$ circuit has minimal $T$-count $1$.
This allows us to approach the question of reducing $T$ count by considering decompositions of unitaries involving few \pipp\ operations, acting on many qubits. 

\vspace*{-2ex}
\paragraph{Previous work on \pipp\ operations, and the role of the ZX calculus.} 
Phase-parity operations were identified early in our work as objects of interest, independently of that of Ref.~\cite{KvdW-2019} (which is also informed by the ZX calculus) or Refs.~\cite{AM-2019,ZhangCheng19} (which do not use the ZX calculus).
Amy and Mosca~\cite{AM-2019} identify these as relevant unitary operators, but immediately proceed to describe them rather in terms of more local controlled-phase operations.
Kissenger and Van\;de\;Wetering~\cite{KvdW-2019} seem to have identified \pipp\ operations (in the form of ZX phase gadgets) for similar reasons to us: there is the sense of being confronted with them as the principal object of study, but the lack of commitment of the ZX calculus to the circuit model allows one to be more relaxed about their nature as many-qubit operations.
We note that Zhang and Chen~\cite{ZhangCheng19}, and for that matter Litinski~\cite{Litinski-2019}, demonstrate that the ZX calculus is not actually required to productively reason about \pipp\ operations.
Indeed, little knowledge to the ZX calculus is required either to understand or to make use of our results.
The role played by the ZX calculus in our work is therefore not an essential one: instead, the role played by the ZX calculus was to quickly single out \pipp\ operations as the relevant objects of study, and to allow us to easily reason about them --- which are the main things that one might reasonably ask of a good mathematical notation.

\vspace*{-1ex}
\section{Reduction of $T$-count through simplification of \pp\ circuits}

In this section, we apply circuit reduction techniques similar to those of Heyfron and Campbell~\cite{HC-2018}, augmented with techniques motivated by the ZX calculus, to describe simplifications which can reduce the $T$-count necessary to realise a unitary $\mathcal D_3^n$ operation.
Our results do not make heavy (explicit) use of the re-write rules of the ZX calculus: a reader who is content with circuits including intermediate measurements, and who is comfortable with reading a parity-phase gadget such as that of Eqn.~\eqref{eqn:diagramD-Sk} as a unitary operator, may interpret every diagram below as a circuit diagram.

\vspace*{-1ex}
\subsection{Reduction to ``homogeneous'' circuits of $T$-gadgets}
\label{sec:reductionToDk}

We first consider a series of circuit transformations, following (and mildly extending) that of Heyfron and Campbell~\cite{HC-2018}, to reduce the amount of non-Clifford diagonal operations used to realise a unitary $U$.
We suppose that $U$ is given by a circuit $\mathbf C$, initially expressed as a circuit over the gate-set $\bigl\{X, \mathrm{CNOT}, \mathrm{CCNOT}, Z, \mathrm{C}Z, \mathrm{CC}Z, H, S, T, \mathrm{SWAP}\bigr\}$ --- of which all  gates apart from $\{\mathrm{CCNOT}, \mathrm{CC}Z, T\}$ are Clifford gates.
We note that reversible circuits commonly involve multiply-controlled-NOT gates with more than two controls: for the sake of simplicity we suppose that these have been decomposed into $\mathrm{CCNOT}$ gates, involving computation and uncomputation on auxiliary qubits initialised to $\ket{0}$ in the usual way (though superior techniques to this are by now well-established: see~\emph{e.g.}~\cite{Jones-2013,Gidney-2018,MSCRdM-2019}).

The main concept is to isolate a ``homogeneous circuit'' of $\mathcal D_3^n$ operations, preceded and followed by circuits consisting entirely of (possibly classically-controlled) Clifford operations and Pauli observable measurements.
To this end, we transform $\mathbf C$ as follows:
\begin{enumerate}
\item
  Replace the reversible classical operations $X$, $\mathrm{CNOT}$, and $\mathrm{CCNOT}$ with the decompositions $H\:\!Z\:\!H$, ${(\idop \!\otimes\! H) \,\mathrm{C}Z\, (\idop \!\otimes \!H)}$, and ${(\idop \!\otimes\! \idop \!\otimes\! H) \,\mathrm{CC}Z\, (\idop \!\otimes\! \idop \!\otimes\! H)}$ respectively.
  --- After this step, $\mathrm{CC}Z$ and $T$ are the only non-Clifford gates in $\mathbf C$.
  
\item
  Cancel consecutive pairs of self inverse gates $H$, $X$, $Z$, $\mathrm{C}Z$, or $\mathrm{SWAP}$ which occur in the circuit (\emph{e.g.},~such as may be introduced in the preceding step), and commute as many Clifford gates to the beginning\;/\;end of $\mathbf C$ as possible without transforming any of the gates in the circuit (\emph{e.g.},~commuting $\mathrm{C}Z$ operations but not Hadamard operations past $\mathrm{CC}Z$ operations).
  We refer to these as the ``initial'' and ``final'' Clifford stages of $\mathbf C$ below, and the rest of $\mathbf C$ as the ``main body''.

\item
  From the earliest $H$ gate in the circuit to the latest, determine whether it can either be commuted to the initial or final Clifford part of the circuit --- or commuted to be adjacent to another $H$ gate on the same qubit --- by suitable transformations of $X$ gates, $Z$ gates, $\mathrm{C}Z$ gates, or the targets of $\mathrm{CNOT}$ gates.
  If it is possible to commute it in this way, do so.

\item
  Repeat step~2 to cancel pairs of $H$ gates, or extract any further Clifford operations, to the initial or final Clifford stages of $\mathbf C$.
\item
  Rewrite the $H$ gates in the interior of the circuit, using the following circuit\,/\,ZX gadgets (using a fresh classical bit label in place of ``$s$'' below, each time):
  \vspace*{-1ex}
  \begin{equation}{}
  \label{eqn:HadamardGadgetCircuit}
  \mspace{-36mu}
  \begin{aligned}
  \begin{tikzpicture}
    \def\dx{0.4}
    \def\dy{0.66}
    \coordinate (0) at (0,0);
    \coordinate (1) at ($(0) + (\dx,0)$);
    \coordinate (2) at ($(1) + (\dx,0)$);
    \draw (0) -- (1) -- (2);
    \node [style=H box] at (1) {$H$};
  \end{tikzpicture}
  \end{aligned}
  \;\;
  \equiv
  \;\;
  \begin{aligned}
  \begin{tikzpicture}
    \def\dx{0.425}
    \def\dy{1.0}
    \coordinate (0) at (2,{\dy/2});
    \coordinate (x0) at (0);
    \coordinate (y0) at ($(x0) + (0,-\dy)$);
    \xdef\u{0}
    \foreach \t in {1,...,5} {%
      \foreach \v in {x,y} {%
        \ifnum\t=1
          \coordinate (\v\t) at ($(\v\u) + ({2.5*\dx},0)$);
        \else\ifnum\t>3
          \coordinate (\v\t) at ($(\v\u) + ({1.5*\dx},0)$);
        \else
          \coordinate (\v\t) at ($(\v\u) + (\dx,0)$);
        \fi\fi
        \draw (\v\u) -- (\v\t);
      }
      \xdef\u{\t}
    }
    \node [fill=white,draw=none,anchor=west] at (y0) {$\ket{\texttt+}$};
    \filldraw [black] (x3) circle (2pt) -- (y3) circle (2pt);
    \draw [white,line width=2pt] (x1) -- (x2); 
    \draw [white,line width=2pt] (y1) -- (y2);
    \draw (y1) .. controls ++(0.125,0) and ++(-0.125,0) .. (x2);
    \draw [white,line width=4pt] (x1) .. controls ++(0.125,0) and ++(-0.125,0) .. (y2);
    \draw (x1) .. controls ++(0.125,0) and ++(-0.125,0) .. (y2);
    \node [draw, fill=white, inner sep=2pt, font=\small\itshape] (condY) at (x4) {$X$};
    \node [draw, fill=white, inner sep=2pt, label distance=-5mm, minimum height=5mm, minimum width=5.5mm] (meas) at (y4) {};
    \draw ($(meas.south) + (-2.125mm,1.5mm)$) arc (150:30:2.5mm); \draw ($(meas.south) + (0,1mm)$) -- ++(2mm,3mm);
    \node [anchor=north west, inner sep=1.5pt, font=\tiny\itshape] at (meas.north west) {X};
    \draw [double] (meas) -- (condY);
    \draw [white, line width=2pt] (meas) -- (y5);
  \end{tikzpicture}
  \end{aligned}
  \;
  \equiv
  \begin{aligned}
  \begin{tikzpicture}[scale=1]
    \def\dx{0.425}
    \def\dy{1.0}
    \coordinate (0) at (2.4,{\dy/2});
    \coordinate (x0) at (6,0);
    \coordinate (y0) at ($(x0) + (0,-\dy)$);
    \xdef\u{0}
    \foreach \t in {1,...,9} {%
      \foreach \v in {x,y} {%
      \ifnum\t<2
          \coordinate (\v\t) at ($(\v\u) + (\dx,0)$);
      \else\ifnum\t<3
          \coordinate (\v\t) at ($(\v\u) + (0.5*\dx,0)$);
      \else\ifnum\t<4
        \coordinate (\v\t) at ($(\v\u) + (\dx,0)$);
      \else\ifnum\t<8
        \coordinate (\v\t) at ($(\v\u) + (0.875*\dx,0)$);
      \else
        \coordinate (\v\t) at ($(\v\u) + (\dx,0)$);
\fi\fi\fi\fi
        \draw (\v\u) -- (\v\t);
      }
      \xdef\u{\t}
    }
    \draw [white, line width=2pt] (y0) -- (y1);
    \filldraw [style=Z dot] (y1) circle (3pt);
    \draw [white,line width=2pt] (x2) -- (x3); 
    \draw [white,line width=2pt] (y2) -- (y3);
    \draw (y2) .. controls ++(0.125,0) and ++(-0.125,0) .. (x3);
    \draw [white,line width=4pt] (x2) .. controls ++(0.125,0) and ++(-0.125,0) .. (y3);
    \draw (x2) .. controls ++(0.125,0) and ++(-0.125,0) .. (y3);
    \filldraw [style=Z dot] (x4) circle (3pt)
      node [anchor=south,inner sep=-1pt,font=\footnotesize] {$-\pi\!\!\:/2$};
    \filldraw [style=Z dot] (y4) circle (3pt)
      node [anchor=north,inner sep=-1pt,font=\footnotesize] {$-\pi\!\!\:/2$};
    \coordinate (xy6) at ($(x6)!0.625!(y6)$);
    \draw (x5) -- (xy6);
    \draw (y5) -- (xy6);
    \filldraw [style=Z dot] (x5) circle (3pt);
    \filldraw [style=Z dot] (y5) circle (3pt);
    \coordinate (xy7) at ($(x7)!0.625!(y7)$); 
    \draw (xy6) -- (xy7);
    \filldraw [style=X dot] (xy6) circle (3pt);
    \filldraw [style=Z dot] (xy7) circle (3pt) node [anchor=-90,inner sep=0pt,font=\footnotesize] {$\!\!\;\pi\!\!\:/\!\!\;2\!\!$};
    \filldraw [style=X dot] (x8) circle (3pt) node [anchor=-120,inner sep=0pt,font=\footnotesize] {$\!\pi\!\!\;,\!\{\!\!\;s\!\!\;\}$};
    \draw [white, line width=2pt] (y8) -- (y9);
    \filldraw [style=Z dot] (y8) circle (3pt) node [anchor=north,inner sep=-1pt,font=\footnotesize] {$\!\pi\!\!\;,\!\{\!\!\;s\!\!\;\}$};
    \end{tikzpicture}
  \end{aligned}  
  \equiv
  \begin{aligned}
  \begin{tikzpicture}[scale=1]
    \def\dx{0.425}
    \def\dy{1.0}
    \coordinate (0) at (2.4,{\dy/2});
    \coordinate (x0) at (6,0);
    \coordinate (y0) at ($(x0) + (0,-\dy)$);
    \xdef\u{0}
    \foreach \t in {1,...,9} {%
      \foreach \v in {x,y} {%
      \ifnum\t<2
          \coordinate (\v\t) at ($(\v\u) + (\dx,0)$);
      \else\ifnum\t<3
          \coordinate (\v\t) at ($(\v\u) + (0.5*\dx,0)$);
      \else\ifnum\t<4
        \coordinate (\v\t) at ($(\v\u) + (\dx,0)$);
      \else\ifnum\t<8
        \coordinate (\v\t) at ($(\v\u) + (0.75*\dx,0)$);
      \else\ifnum\t<9
        \coordinate (\v\t) at ($(\v\u) + (2*\dx,0)$);
      \else
        \coordinate (\v\t) at ($(\v\u) + (1.25*\dx,0)$);
\fi\fi\fi\fi\fi
        \draw (\v\u) -- (\v\t);
      }
      \xdef\u{\t}
    }
    \draw [white, line width=2pt] (y0) -- (y1);
    \filldraw [style=Z dot] (y1) circle (3pt)
      node [anchor=north,inner sep=1pt,font=\footnotesize] {$\!\!\!-\!\!\:\pi\!\!\:/2$};
    \draw [white,line width=2pt] (x2) -- (x3); 
    \draw [white,line width=2pt] (y2) -- (y3);
    \draw (y2) .. controls ++(0.125,0) and ++(-0.125,0) .. (x3);
    \draw [white,line width=4pt] (x2) .. controls ++(0.125,0) and ++(-0.125,0) .. (y3);
    \draw (x2) .. controls ++(0.125,0) and ++(-0.125,0) .. (y3);
    \coordinate (xy5) at ($(x5)!0.625!(y5)$);
    \draw (x4) -- (xy5);
    \draw (y4) -- (xy5);
    \filldraw [style=Z dot] (x4) circle (3pt);
    \filldraw [style=Z dot] (y4) circle (3pt);
    \coordinate (xy6) at ($(x6)!0.625!(y6)$); 
    \draw (xy5) -- (xy6);
    \filldraw [style=X dot] (xy5) circle (3pt);
    \filldraw [style=Z dot] (xy6) circle (3pt) node [anchor=-90,inner sep=0pt,font=\footnotesize] {$\!\!\;\pi\!\!\:/\!\!\;2\!\!$};
    \filldraw [style=X dot] (x8) circle (3pt) node [anchor=-120,inner sep=0pt,font=\footnotesize] {$\!\pi\!\!\;,\!\{\!\!\;s\!\!\;\}$};
    \filldraw [style=Z dot] (y7) circle (3pt)
      node [anchor=north,inner sep=-1pt,font=\footnotesize] {$\!\!\!-\pi\!\!\:/2$};
    \draw [white, line width=2pt] (y8) -- (y9);
    \filldraw [style=Z dot] (y8) circle (3pt) node [anchor=north,inner sep=-1pt,font=\footnotesize] {$\!\pi\!\!\;,\!\{\!\!\;s\!\!\;\}$};
    \end{tikzpicture}
  \end{aligned}  
  \mspace{-12mu}
  \end{equation}
  --- After this step, $X$ and $\mathrm{CNOT}$ are the only non-diagonal gates left in the main body of $\mathbf C$.
  \medskip  

  \textbf{Interpretive remark.}
  In the circuit second from the left, the two qubits are subject to a SWAP operation, followed by a $\mathrm CZ = \exp\bigl(i\pi\ket{11}\!\!\bra{11}\bigr)$ 
  operation.
  The bottom qubit is measured finally with an $X$ observable measurement (\ie,~in the $\ket{\texttt{\rlap{\raisebox{0.25ex}{+}}{\raisebox{-0.5ex}{-}}}}$ basis), and the top operation is acted on finally by an $X$ operation only if the outcome is $\ket{\texttt-}$.
  On the right are annotated ZX diagrams in the style of Ref.~\cite{DP-2010}, in which the measurement is represented as a projection with a random outcome $s$ which is heralded and may be used to control  phase operations elsewhere.
  The leftmost ZX diagram describes the decomposition of the controlled-$Z$ operation, using $\mathrm CZ_{h,j} \raisebox{0.25ex}{${}\propto D_{\!\!\:\{h,j\},2}^{\phantom{\mathrlap{-1}}} \;\! D_{\!\!\:\{h\},2}^{-1} \;\! D_{\!\!\:\{j\},2}^{-1}$}$\;.
  The final ZX diagram propogates the single-qubit \smash{$D_{\{\ast\},2}^{-1}$} operations towards the preparation and measurement of the second qubit, so that the second qubit is prepared in the $\ket{\texttt{-y}} \propto \ket{0} - i \ket{1}$ state%
  .

\item
  Decompose $\mathrm{CC}Z$ operations in $\mathbf C$ using the formula of Eqn.~\eqref{eqn:decomposeCSandCCZ}, and represent $T$ gates (on some qubit $j$) by $D_{\{j\},3}$.
  --- After this, all non-Clifford operations are \pipp\ operations, and are in the main body of $\mathbf C$.
  
\item  
  Commute all remaining $X$, $Z$, $\mathrm{CNOT}$, $\mathrm{C}Z$, and $\mathrm{SWAP}$ gates to the beginning or end of $\mathbf C$, out of the main body and into the initial of final Clifford phases.
  This may transform various $D_{\!\!\:S,\!\;t}$ gates by Eqns.~\eqref{eqn:conjugate-Dk-NOTs}, changing the set $S$ involved and/or negating the phase, according to the following commutation relations:
  
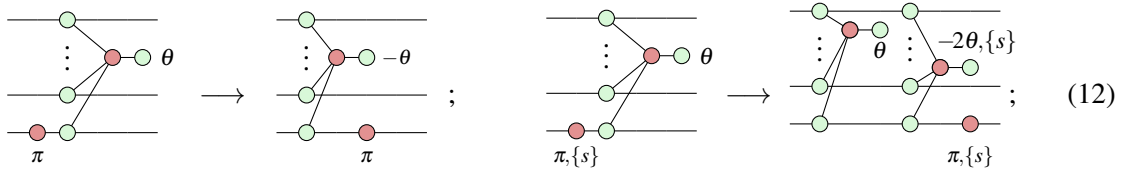
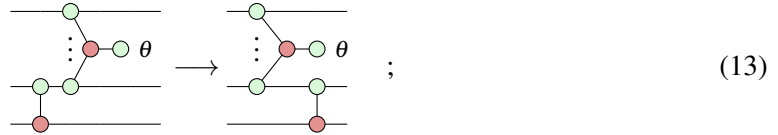
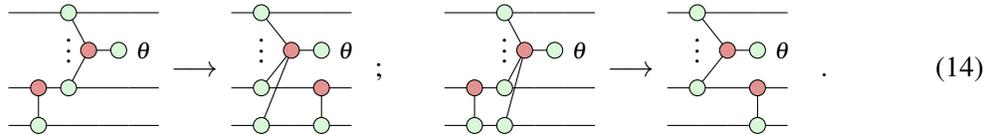
\vspace*{1ex}
  \begin{equation}{}
    \label{x-commute-phgadget}
    %
          \input{./diagrams/commute-x-gadget.tikz}
  \!\!\!\!;\qquad\quad
    %
          \input{./diagrams/commute-cond-x-gadget.tikz}
  \;;
  \end{equation}

  \begin{equation}{}
    \label{cx-commute-phgadget-a}
    %
          \input{./diagrams/commute-cx-gadget-a.tikz}
  \!\!\!\!;
  \end{equation}

  \bigskip
  \begin{equation}{}
    \label{cx-commute-phgadget-b}
    %
          \input{./diagrams/commute-cx-gadget-b.tikz}
  \!\!\!\!\!\!.
  \end{equation}

  \medskip
  --- After this, the main body of $\mathbf C$ will be a commuting circuit, consisting entirely of \pipp~operations.

\item
  Fuse together any $D_{\!\!\:S,k}$ operations for $k \leqslant 3$ which arise on a common subset $S$:
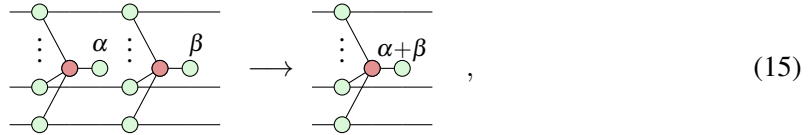
    \begin{equation}{}
    \label{gadget-fuse}
    %
          \input{./diagrams/gadget-fuse.tikz}
  \!\!\!\!\!\!,
  \end{equation}
  and apply Eqn.~\eqref{eqn:phaseParityGateIdentity} if possible to reduce or eliminate these gadgets when possible (in particular, removing entirely any operations $D_{S,k}$ for $k \leqslant 0$).
  Commute any  operations $D_{S,k}$ for $k \in \{1,2\}$ to the final Clifford phase of $\mathbf C$.

\item
  Commute any classically-controlled $\mathcal D_{\!\!\:S,\!\;t}$ gates to the beginning of the final Clifford phase of $\mathbf C$,  in layers according to the classical control bit involved.
\end{enumerate}
If the original circuit $\mathbf C$ had $m$ Hadamard gates, the above procedure realises a transformation
\vspace*{-.5ex}
\begin{equation}
    \mathbf C
  \;\longrightarrow\;
    \mathbf{Cl}_1 \, \mathbf D_m \cdots \,\mathbf D_2 \,\mathbf D_1 \,\mathbf D_0 \, \mathbf {Cl}_0\;,
\end{equation}
\vspace*{-3ex}

\noindent
where (reading right-to-left) $\mathbf {Cl}_0$ and $\mathbf{Cl}_1$ are the initial and final Clifford phases of the circuit respectively; $\mathbf D_0$ is a circuit realising a $\mathcal D_3^n$ operation; and the circuits $\mathbf D_j$ (for $1 \le j \le m$) consist of
the $j\textsuperscript{th}$ measurement in the $\ket{\texttt{\raisebox{0.25ex}{+}\llap{\raisebox{-0.45ex}{-}}}}$  basis with outcome $s_{\!j}$ (denoted in ZX by a green ``$\pi,\!\!\:\{s_j\}$'' node), followed by $\mathcal D_k^n$ operations conditioned on the outcome $s_{\!j}$\,.
We refer to a circuit with this structure as Cl-D-Cl (pronounced ``cliddicle'') form.

In the circuits produced by this procedure, all of the non-Clifford operations are $\mathcal D_3^n$ operations in $\mathbf D_0$.
In particular, each of them is a \pipp\ operation $D_{\!\!\:S,\!\;3}$ --- which can be realised by CNOT gates and a single $T$ gate.
This motivates the question of how to simplify a circuit consisting entirely of \pipp\ operations.
In some instances, we find a signifcant reduction in the $T$-count simply by representing the contributions to the $T$-count entirely in terms of \pipp\ operations, and ``fusing'' these operations together using $D_{\!\!\:S,3}^2 = D_{\!\!\:S,2}$ for any subset $S \subseteq [n]$.
However, in general it is useful to consider what other reduction techniques can be used to simplify homogeneous circuits of \mbox{``$T$-gadgets''}.
In the setting of simplifying such a homogeneous circuit, we may easily make used of the TODD subroutine of Heyfron and Campbell~\cite{HC-2018}; to this we add another technique which \textbf{(a)}~in some cases yields $T$-counts which are better than any previously known results, whether or not one also uses TODD; and \textbf{(b)}~is in principle extensible, allowing for the possibility of further improvements through improved algorithms for this sub-problem.

\vspace*{-1ex}
\subsection{Phase Gadget Elimination tactics}
\label{sec:PHAGE}

Reducing the $T$-count while preserving the meaning of a circuit, implicitly involves applying a mathematical identity, possibly passing temporarily through different representations of these circuits.
(These are often identities of diagonal unitary circuits~\cite{AMM-2014,AM-2019}, though not always~\cite{GKMR-2014,KvdW-2019}.)
In the special case of reductions by identities of \pipp\ operations, these  may in principle be described in terms of a commuting product of operations which are proportional to the identity operator.
We now consider a general approach to the reduction of such circuits by an analysis of families of non-trivial circuits which realise the identity transformation.

\subsubsection{PHAGE tactics}
\label{sec:PHAGE-tactics-in-general}

In the following, we use the terms ``identity of \pipp\ operations'' or ``identity of phase gadgets'' (or simply ``an identity'') to refer to a circuit $\mathcal J$, whose $T$-count is at least $1$ but which nevertheless realises the identity operation.
We make the simple observation that for any family $\mathcal F$ of such ``identities'', there is an associated ``tactic'' to reduce the $T$-count in a homogeneous circuit $\mathbf C$ of such phase gadgets.
For a given subset $S \subseteq [n]$, this tactic is as follows:
\begin{enumerate}
\item
  Determine whether there is an identity $\!\!\mathcal J \in \mathcal F$, such that $\mathbf C$ contains at least half of the $T$-gadgets (or alternatively the inverses of $T$-gadgets) which occur in $\!\!\mathcal J$.
\item
  For any such identity $\!\!\mathcal J$, compute a circuit $\mathbf C_{\!\!\!\mathcal J}$ as the product of $\mathbf C$ and $\mathcal J^{-1}$ (simplifying this circuit by fusing phase gadgets, possibly cancelling $T$-gadgets or otherwise turning into Clifford gadgets.
  Determine the resulting $T$-count.
\item
  Replace $\mathbf C$ with the circuit $\mathbf C_{\!\!\!\mathcal J}$ with the smallest $T$-count, if this is less than the $T$-count of $\mathbf C$ itself.
\end{enumerate}
We call such a procedure a ``Phase Gadget Elimination'' (or PHAGE) tactic.
This procedure is apparently ``greedy'', in that it selects the circuit $\mathbf C_{\!\!\!\mathcal J}$ which minimises the $T$ count after a single application.
It is possible to take a subtler view, in which the family $\mathcal F$ of identities which may be deployed is only implicitly defined, in a way which may depend on the particular structure of $\mathbf C$ or how it acts on $S$.
The main principle of a PHAGE tactic is in that it selects a particular way to reduce the $T$ count based on the independent comparison of one or more different possible identities after some bounded-time procedure.

In principle, the {T-optimise} subroutine of Ref.~\cite{AM-2019} and the TOOL and TODD subroutines of Ref.~\cite{HC-2018} may be interpreted as algorithms to deploy one or more PHAGE tactics, possibly more than once in sequence, and possibly with a random choice of family $\mathcal F$ (and where $\mathcal F$ itself may on occasion be a singleton set).
This approach to $T$-count reduction can be distinguished from that of Ref.~\cite{KvdW-2019}, in which phases may in principle be aggregated one at a time in circuits which are unitary but not diagonal.
While such techniques seem fruitful, we suggest that investigation of identities on \pp\ operations --- and the way in which such identities may be deployed as PHAGE tactics --- may provide a complementary approach to reduce the $T$-count.

The difficulty in reducing the $T$-count arises from the fact that there are a very large number of identities of \pipp\ operations, and a large number of subsets $S \subseteq [n]$ which one may consider.
A na\"{\i}ve approach is simply to let $\mathcal F$ be the family of all identities on $n$ qubits: as this set is exponentially large, the associated PHAGE tactic is infeasible to use for large $n$.
The difficulty is in formulating a successful \emph{strategy}, in which one selects a more appropriately-sized family $\mathcal F$ of identities to try on a particular circuit or subsystem $S$.
The question is then one of having a range of tactics which one may efficiently explore, and also successfully deploy, to reduce the $T$-count.

\subsubsection{Spider nest identities}

Our results depend on a PHAGE tactic --- \emph{i.e.},~an~approach to attempt to reduce homogeneous circuits of phase gadgets --- which is induced by a simple family of identities of \pipp\ operations, which we now describe.
While these identities are in a sense elementary, to our knowledge they have not previously been noted in the literature (though Maslov and Roetteler~\cite[Theorem~2]{MR-2018} make similar observations for operations in $\mathcal D_2^n$).

The identities can be composed from some specific homogeneous circuits which realise the identity (essentially a set of generators for the group of functions $\mathcal C_n$ described by Amy and Mosca~\cite{AM-2019}), which involve a single $T$-phase $n$-gadget for $n \geq 4$, and phase $k$-gadgets with $k \leq 3$: 
\vspace*{1ex}
\begin{equation}
  \label{gadgetdec}
  %
          \input{./diagrams/tphasegadgetdecom.tikz}
 .
\end{equation}
\vspace*{-1ex}

\noindent
Let $G_n$ denote the $n$-qubit circuit on the left-hand side of Eqn.~\eqref{gadgetdec}.
This consists of a $1$-gadget with phase angle $(n{-}2)(n{-}3)\tfrac{\pi}{8}$ on each line, a $2$-gadget on each pair of lines with phase angle $-(n{-}3)\tfrac{\pi}{4}$, and a $3$-gadget with phase angle $\tfrac{\pi}{4}$ on each subset of three lines, and finally an $n$-gadget with phase angle $-\tfrac{\pi}{4}$.
(We prove this identity in Appendix~\ref{sec:n-gadgetDecomposition}.)

We refer to identities of the form of Eqn.~\eqref{gadgetdec} --- and any other identity involving a small number of large phase-gadget ``spiders'' together with a large number of smaller phase-gadget ``spiders'' --- as a \emph{spider nest} identity.

\vspace*{-1ex}
\paragraph{Features of simple spider nest identities.}
 
Our results in fact make only limited (but crucial) use of spider nest identities.
As it seems likely to us that these identities can be used to greater effect than we have in our results, we now describe some features of these identities in general.
Let $\mathcal N_S$ represent the homogeneous circuit of phase gadgets on the left-hand side of Eqn.~\eqref{gadgetdec}, acting on a set $S = \{1,2,\ldots,n\}$ of cardinality $n$.
Note the following features of $\mathcal N_S$:
\begin{itemize}
\item  
  If $n=4$, then it is essentially the same as the rule $R_{13}$ given in  \cite{ACR-2018}, and also Eqn.~\eqref{eqn:phaseParityGateIdentity}.
\item
  If $n \equiv 1\!\!\pmod{4}$ or $n \equiv 3 \!\!\pmod{4}$, all of the $2$-gadgets in Eqn.~\eqref{gadgetdec} are Clifford-phase gadgets, which do not contribute to the $T$-count.
\item
  If $n \equiv 3 \!\!\pmod{4}$ or $n \equiv 2 \!\!\pmod{4}$, all of the $1$-gadgets in Eqn.~\eqref{gadgetdec} are Clifford-phase gadgets, which again do not contribute to the $T$-count.
\end{itemize}
\noindent
For a fixed value of $n$, and a $T$-phase gadget on $1$ to $3$ qubits, there is a question of whether or not such a gadget is involved in $\mathcal N_S$\,, as a number of the phase gadgets involved are Clifford gadgets instead.
This would affect the number of phase gadgets involved in the identity, and therefore in a sense how easily one may find subcircuits on which the induced PHAGE tactic could be fruitfully used.
Let $\mathbf T_n$ denote the $T$-count of $\mathcal N_S$\,: then 
\begin{equation}
  \label{tcount}
  \mathbf T_n
=
  \begin{cases}
    \tfrac{1}{6}n(n^2+5),       & \text{for $n \equiv 0 \pmod{4}$}; \\[0.5ex]
    \tfrac{1}{6}n(n^2-3n+8),    & \text{for $n \equiv 1 \pmod{4}$}; \\[0.5ex]
    \tfrac{1}{6}n(n^2-1),       & \text{for $n \equiv 2 \pmod{4}$}; \\[0.5ex]
    \tfrac{1}{6}n(n^2-3n+2),    & \text{for $n \equiv 3 \pmod{4}$}.
  \end{cases}
\end{equation}
In some cases, it is also possible to compose two or more circuits $\mathcal N_{S_j}$ (or their inverses) to obtain a ``composite'' spider nest identity which has a small $T$-count.
This may be helpful for finding simplifications of circuits, through PHAGE tactics which use such an identity.
For instance, consider the specific circuit $\mathcal N_S \, \mathcal N_{S'}^{-1}$ where $\lvert S \rvert \geqslant 5$ and $S' = S \setminus\! \{ r \}$ for some $r \in S$.
In this composite circuit, various small $T$-gadgets of $\mathcal N_{S'}^{-1}$ and of $\mathcal N_S$ cancel each other out, yielding a circuit of the form:
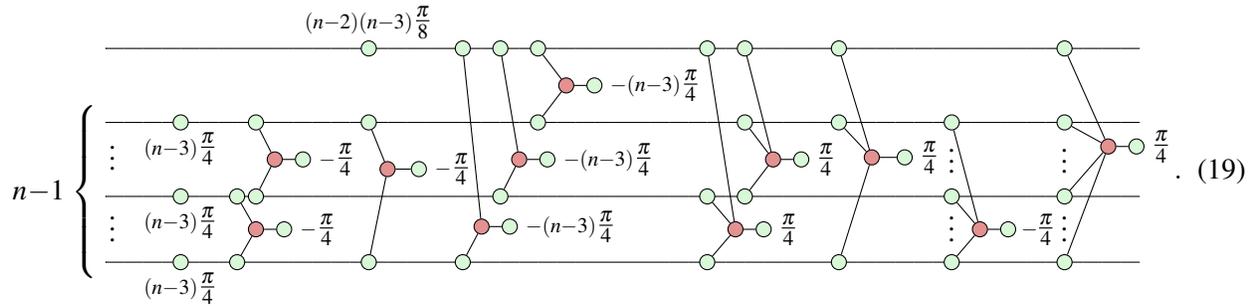
\vspace*{1ex}
\begin{equation}{}
  \label{2gadgetdec}
  \mspace{-48mu}
  %
          \input{./diagrams/order2-spidernest.tikz}
\!\!.
  \mspace{-12mu}
\end{equation}
Let $r = S \setminus S'$ represent the top qubit in the diagram above.
The purpose of composing $\mathcal N_{S}$ with $\mathcal N_{S'}^{-1}$ is to fuse the various phase gadgets together (as we have above) to cancel the majority of the $3$-gadgets on $S$ against the $3$-gadgets on $S' \prsubset S$, and potentially to cancel almost all of the $1$-gadgets on $S$ as well. 
We are then left with whichever $1$- and $2$-gadgets on $S'$ are left uncanceled, a collection of gadgets from $\mathcal N_S$ involving $r$ interacting with some one or two qubits of $S'$, and the large phase gadgets acting on all of $S'$ and $S$ respectively.
If $\tilde{\mathbf T}_n$ denotes the $T$-count of the circuit above, we then have
\begin{equation}
  \label{tcount2}
  \tilde{\mathbf T}_n
=
  \begin{cases}
    n^2 - n + 2 + \delta_n & \text{for $n$ even}; \\[0.5ex]
    n^2 - 3n + 4 + \delta_n           & \text{for $n$ odd},
  \end{cases}
\end{equation}
where $\delta_n = 1$ if $n \equiv 0$ or $n \equiv 1$ modulo $4$, and $\delta_n = 0$ if $n \equiv 2$ or $n \equiv 3$ modulo $4$ (determining whether the $1$-gadget on qubit $r$ has $T$-count one or zero).
As this is asymptotically smaller than $\mathbf T_n$, one may see how it could be easier to find scenarios in which a single application of the identity $\mathcal N_S \mathcal N_{S'}^{-1} \propto \idop$ is beneficial.
(This most be weighed against the prospect that the asymmetry between the qubits in Eqn.~\eqref{2gadgetdec} will imply that the structure of the input circuit will play a role in whether it is likely to be useful.)

\subsubsection{Two na\"{\i}ve spider-nest PHAGE tactics}

We now describe the way in which we use spider nest identities to obtain our results.
This involves two simple PHAGE tactics,  relative to the scheme described in Section~\ref{sec:PHAGE-tactics-in-general}, using distinct families $\mathcal F_1, \mathcal F_2$ of spider nest identities.

\medskip
\begin{tactic*}[``STOMP~4'']
  For a set $S = \{q_1, q_2, q_3, q_4\}$, apply the PHAGE tactic associated with the family $\mathcal F_1 = \{ \mathcal N_S \}$ on the set $S$.
\end{tactic*}

\smallskip
\begin{tactic*}[``STOMP~5'']
  For a set $S = \{q_1, q_2, q_3, q_4, q_5\}$, apply the PHAGE tactic associated with the family $\mathcal F_2$ consisting of the $63$ different identities $\mathcal N_S^{p_0} \mathcal N_{S_1}^{p_1} \mathcal N_{S_2}^{p_2} \mathcal N_{S_3}^{p_3} \mathcal N_{S_4}^{p_4} \mathcal N_{S_5}^{p_5}$, where we define $S_j = S \setminus \{ j \}$ for each $1 \leqslant j \leqslant 5$, and where $p_0p_1p_2p_3p_4p_5 \in \{0,1\}^6$ is not all zero.
\end{tactic*}

\smallskip
\noindent
These PHAGE tactics do not exploit many properties of the spider nest identities described above: they consist of a brute-force application of (possibly composite) spider nest identities on small subsystems.
The tactic STOMP~5 in particular is motivated in part by the lower $T$-count involved in the composite spider nest identity of Eqn.~\eqref{2gadgetdec}: by testing many such composites, we attempt to find local opportunities to reduce the $T$-count.
The strategies which use use to deploy them are also very simple: on any homogenous circuit $\mathbf C$ of \pipp\ operations on $n$ qubits, first we apply STOMP~4 on all subsets of size $4$ in some order, and then we apply STOMP~5 on all subsets of size $5$ in some order.
(This is somewhat redundant, as $\mathcal F_2$ contains five different identities $\mathcal N_{S_j}$ for $1 \le j \le 5$; we may simplify this by requiring that $p_0 p_1 p_2 p_3 p_4 p_5$ have Hamming weight at least $2$, replacing $\mathcal F_2$ with a family $\mathcal F'_2$ of a mere 58 identities.)

As we show in Section~\ref{sec:results}, in many cases we obtain the best known $T$-count for a number of circuits.
Even so, our result may be considered only a proof of principle of the usefulness of spider nest identities --- a more sophisticated application of them may well yield superior results to those we present below.

\vspace*{-1ex}
\subsection{Analysis of a procedure to reduce $T$-count}

We now describe the reduction procedure used in our results.
Suppose that we are given a circuit $\mathbf C$ on $n$ qubits, over the gate-set  $\bigl\{X, \mathrm{CNOT}, \mathrm{CCNOT}, Z, \mathrm{C}Z, \mathrm{CC}Z, H, S, T, \mathrm{SWAP}\bigr\}$.

\subsubsection{The reduction procedure }

We perform the following transformations on $\mathbf C$.
\def\customenumi{\arabic{enumi}.\phantom*}
\begin{enumerate}[label=\protect\customenumi]%
\item
  Reduce the circuit $\mathbf C$ to a Cl-D-Cl form, using the procedure described in Section~\ref{sec:reductionToDk}.
  This serves to isolate a homogeneous circuit of commuting \pipp\ operations, with the rest of the circuit consisting of Clifford group operations (possibly conditioned on the outcomes of measurements).
  This yields a circuit $\mathbf C'$ on $N \geqslant n$ qubits.
\item
  Perform the  PHAGE tactic STOMP~4 on all subsets of size $4$, in some sequence, from the $N$ qubits on which $\mathbf C'$ acts.
  Call the resulting circuit $\mathbf C''$.
\item
  Perform the PHAGE tactic STOMP~5  on all subsets of size $5$, in some sequence, from the $N$ qubits on which $\mathbf C''$ acts.
  Call the resulting circuit $\mathbf C'''$.
  \def\customenumi{\arabic{enumi}*.}
\item
  Perform TODD on $\mathbf C'''$ some constant number of times, independently; output the circuit which has the smallest $T$-count from among these three runs.
  (Our results used the best outcome from 3--40 independent executions of TODD for each circuit.)
\end{enumerate}

Note that $N$, the number of qubits of the circuit produced as output, is a function of how many Hadamard gates are either involved in $\mathbf C$ or are introduced from the decomposition of CCNOT gates.
More precisely, it also depends on how many of these gates can be commuted from the ``main body'' of $\mathbf C$ to the initial or final Clifford stages.
Thus, for a circuit consisting of $M$ gates, a bound which is substantially better than $N \le n + M$ will be difficult without some knowledge of the structure of $\mathbf C$.
In several cases, we find that many or all of the Hadamard gates introduced by decomposing CCNOT gates can be eliminated: so, $N \le n + M$ is likely a loose upper bound in practical circumstances.

\subsubsection{Remarks on the TODD subroutine}

The final step involves the subroutine TODD described by Campbell and Heyfron~\cite{HC-2018}, for the simple reason that this subroutine is effective at reducing $T$-count without impacting the asymptotic run-time of our algorithm.
It also allows us to demonstrate how using our techniques in conjunction with TODD in some cases yields a result which is better than those found to date using TODD alone.

Heyfron and Campbell bound the run-time of the TODD subroutine as $O(N^3 t^2 + N t^3)$ for $t$ the initial $T$-count of the circuit --- see Ref.~\cite[Eqn.~53]{HC-2018}.
The number of times TODD is invoked for a given circuit is somewhat arbitrary.
As it is a randomised algorithm, it will yield different results in different invocations; and as it is difficult to determine when one has obtained a circuit with optimal (or approximately optimal) $T$-count, one might imagine in principle that running it a larger number of times might eventually yield a better result.
As we show below, the run-time analysis of our algorithm would not be affected were we to run TODD for each circuit $O(\log M)$ times; in practise we contented ourselves with at most 40 times, and in fact at most 3 times for each circuit.

\subsubsection{Run-time analysis}

Our procedure runs in time polynomial in the number of gates $M$ of the input circuit, and can be realised in a run-time which is only slightly larger than the asymptotic upper bound of the TODD subroutine.\footnote{%
  N.B.~The account below of the run-time of our techniques differs from the run-time of the implementation which we used in practise, which used a somewhat less efficient means of applying spider-nest identities.
}
We may describe the asymptotic run-time of each of the steps of our procedure, as follows:
\vspace*{.375ex}
\begin{itemize}
\item  
  Step~1 involves operations which involve simple decompositions of gates, or commutations of pairs of gates, in the circuit, and so runs in time $O(M^2)$.
  As a part of this run-time cost (in time $O(t \log t)$), we may create a tree structure (with $t$ elements) storing the $T$-gadgets in the homogeneous circuit.
\item
  Steps~2 and~3 involve determining whether an identity on $4$- or $5$-qubit subsystems of $N$ qubit homogeneous circuits lead to $T$-count reductions.
  As each identity has constant size, the run-time for this is governed by the number of such subsystems, times the search time for a tree of size $t$, or $O(N^4 \log t)$ and $O(N^5 \log t)$ respectively, where $t$ is the initial $T$-count of the circuit.
\item
  Finally, TODD runs in time $O(N^3 t^2 + N t^3)$.
\end{itemize}
\vspace*{.375ex}

\noindent Thus, our procedure runs in time $O(M^2 + N^5 \log t + N^3 t^2 + N t^3)$.

Consider a family of circuits $\{ \mathbf C_n \}_{n \in \mathbb N}$, with at least one operation on each qubit (so that $M \geqslant \tfrac{1}{3}n$), and in which some constant fraction of the gates of $\mathbf C_n$ are CCNOT gates, whose decomposition in Step~1 introduces Hadamard gates.
Then we have $N = n + \alpha M$ for some $0 \leqslant \alpha \leqslant 2$, and $t = \beta M$ for $0 \leqslant \beta \leqslant 7$.
Our procedure then runs in time $O(M^5 \log M)$, which is dominated by the asymptotic upper bound on the run-time of STOMP~5, and  up to a log-factor is the same as the bound on the run-time of Step~4 (which applies TODD).
\begin{table}[p]
  \catcode`\^=13
  \def^{\textsuperscript}
  \catcode`\_=13
  \def_{\textsubscript}
  \def\best#1!{\textbf{\color{red!40!black}\llap{*\;\!}{#1}}}
  \def\tie#1.{\textbf{\color{green!40!black}#1}}
  \def\wow{\rlap{\textbf{\smash{\,(!)}}}}
  \def\TOpt{\textsf{\footnotesize TOpt\,\;}}
  \def\PyZX{\textsf{\footnotesize PyZX\,\;}}
  \def\TPar{\textsf{\footnotesize TPar\,\;}}
  \def\RMr{\textsf{\footnotesize RM\textsubscript{r}\,\,\;}}
  \begin{center}
  \hspace*{-5em}
  \begin{tabular}[c]{||l|rr@{\;\;\;}|r@{\hspace{2ex}}l@{\quad}r@{\hspace{2ex}}l|rrr@{}r||}
      \hspace{4em} & \;\;\; & \;\;\; & \quad & \quad & \quad & \quad & \phantom{M} & \phantom{MMM}  & \phantom{MMM} & \phantom{M} \\[-3ex]
    \hline \hline
    \begin{minipage}[t]{5em}%
    \bfseries\centering
    \vspace*{3ex}
    Circuit
    \\[-2ex]
    \end{minipage}
  &
    \begin{minipage}[c]{4em}%
    \vspace*{.25ex}%
    \bfseries\centering
    \# extra~~ \\[-0.35ex] qubits
    \vspace*{-3ex}%
    \end{minipage}
    \hspace*{-3.125em}
  &&
    \begin{minipage}[t]{15em}%
    \vspace*{.25ex}%
    \bfseries\centering
    Best prior results
    \end{minipage}
    \hspace*{-15em}
  &&&&
    \begin{minipage}[t]{8em}%
    \vspace*{.25ex}%
    \bfseries\centering
    \hspace*{-5ex}
    Effect of our techniques
    \hspace*{-7ex}
    \end{minipage}
    \hspace*{-7.25em}
  &&&
  \\[-.375ex]
  & \,~Ref.~\,\hspace{-1.5ex} & \,~our~\,\hspace{-1ex} &
    \begin{minipage}[t]{6.5em}%
    \flushleft
    without TODD  
    \end{minipage}
    \hspace*{-5.5em}
  &&
    \begin{minipage}[t]{6.5em}%
    \flushleft
    with TODD  
    \end{minipage}
    \hspace*{-5.5em}
  &&
    \hspace*{-1ex}
    \begin{minipage}[t]{3em}%
    \centering
    Gadget \\[-0.35ex] Fusion  
    \end{minipage}
    \hspace*{-3ex}
  &
    \hspace*{-1ex}
    \begin{minipage}[t]{3em}%
    \centering
    STOMP \\[-0.35ex] 4 \& 5  
    \end{minipage}
    \hspace*{-3ex}
  &
    \hspace*{-1ex}
    \begin{minipage}[c]{3em}%
    \vspace*{1.25ex}%
    \centering
    TODD 
    \end{minipage}
    \hspace*{-3ex}
  &
  \\[-2ex]
  &
    \cite{HC-2018}\hspace{-0.5ex}
  &
    work\hspace{-1ex}
  &
      \#$T$ & algorithms 
  &
      \#$T$ & algorithm
  &
  &
  &
  &
  \\
  & && && && &&&
  \\[-2ex]
    \hline
    Barenco\;Toff_3 & 3  & 3  & 16  & \TPar \TOpt \PyZX & 14  & \TOpt & \tie  16. & \best  13\wow! & \best 13! &\\ 
    Barenco\;Toff_4 & 7  & 7  & 26  & \TOpt             & 24  & \TOpt &       28  & \best  24\wow! & \tie  24. &\\
    Barenco\;Toff_5 & 11 & 11 & 40  & \TPar \PyZX       & 34  & \TOpt & \tie  40. & \best  36!     & \tie  34. &\\
    \hline
    NC Toff_3       & 2  & 2  & 14  & \TOpt             & 13  & \TOpt &       15  & \best  13\wow! & \tie 13. &\\
    NC Toff_4       & 4  & 4  & 22  & \TOpt             & 19  & \TOpt &       23  & \best  19\wow! & \tie 19. &\\
    NC Toff_5       & 6  & 6  & 29  & \TOpt             & 25  & \TOpt &       31  & \best  26!     &      26 &\\
    NC Toff_{10}    & 16 & 16 & 65  & \TOpt             & 55  & \TOpt &       71  & \best  58!     &      56 &\\
    \hline
    GF(2^4)-mult    & 7  & 0  & 68  & \TPar \PyZX       & 52  & \PyZX & \tie  68. & \best  61!     & \best  47! &\\
    GF(2^5)-mult    & 9  & 0  & 101 & \RMr              & 86  & \PyZX &       115 & \best  97!     & \best  84! &\\
    GF(2^6)-mult    & 11 & 0  & 144 & \RMr              & 122 & \PyZX &       150 & \best 134!     & \best 118! &\\
    GF(2^7)-mult    & 13 & 0  & 208 & \RMr              & 173 & \PyZX &       217 & \best 192!     &       175 &\\
    GF(2^8)-mult    & 15 & 0  & 237 & \RMr              & 214 & \PyZX &       264 &       247      &       229 &\\
    \hline
    CSLA-Mux_3      & 17 & 6  & 58  & \RMr              & 45  & \PyZX &       62  & \best  48!     & \best 40! &\\
    HWB_6           & 24 & 20 & 71  & \TPar             & 51  & \TOpt &       75  & \best  62!     &       52  &\\
    Mod5_4          & 6  & 0  & 8   & \PyZX             & 7   & \PyZX & \tie   8. & \best   7\wow! & \tie   7. &\\
    Mod-Mult_{55}   & 10 & 3  & 19  & \TOpt             & 17  & \TOpt &       35  &        26      &       18  &\\
    Mod-Red_{21}    & 17 & 17 & 68  & \TOpt             & 55  & \TOpt &       73  & \best  63!     & \tie  55. &\\
    RC-Adder_6      & 21 & 10 & 44  & \TOpt             & 37  & \TOpt &       47  & \best  39!     & \tie  37. &\\
    VBE-Adder_3     & 4  & 4  & 24  & \TPar \TOpt \PyZX & 20  & \TOpt & \tie  24. & \best  20\wow! & \tie  20. &\\
    \hline \hline
  \end{tabular}
  \hspace*{-5em}
  \end{center}
\caption{%
  \label{table:results}
  Comparison of our techniques for $T$-count reduction against previous techniques, for a selection of benchmark circuits.
  For each circuit, we describe the number of qubits introduced by our algorithm, and the T-counts realised after each stage of our procedure (gadget fusion, then the STOMP PHAGE tactics, and finally the TODD subroutine of Ref.~\cite{HC-2018}).
  We compare the number of additional qubits required to the results of Ref.~\cite{HC-2018}, and we compare our results for $T$-count to the best known prior results.
  The prior results are classified into results which use the (computationally expensive) TODD subroutine of Ref.~\cite{HC-2018}, and those that don't.
  We indicate the algorithms which achieve these results by \TPar\!\cite{AMM-2014}, \TOpt\!\cite{HC-2018} (specifically either TOOL(F), TOOL(NF), or TODD), recursive Reed-Muller decoding \RMr\!\!\cite{AM-2019}, or \PyZX\!\cite{KvdW-2019}.
%
  In each case, we compare the counts achievable after Steps~1 and~3 of our algorithm to the prior results without TODD, and the count achievable after Step~4 to the prior results with TODD.
  ---
  In a number of instances, our results match or improve upon the best previously known results.
  Circuits for which our techniques are the same as or better than the best previous result are in bold-face; those where our results are strictly better are also marked with an asterisk.
  In some instances, we manage to obtain the best known result even without the use of the TODD subroutine, indicated by a (!) mark.
  Note that even when we do not achieve the best known result, we often exceed that result by a single $T$ gate.
}
\label{table:results}  
\end{table}
\vspace*{-1ex}
\section{Results}
\label{sec:results}

Table~\ref{table:results} presents a comparison of the results of our algorithm, with the previous best algorithms for reducing $T$-count.
In order to separately demonstrate the effectiveness of the fusion of phase gadgets, the PHAGE tactics, and TODD, we describe the $T$-count obtained by each of these stages of the algorithm.
Our results do not include an account of the cost of the Clifford group operations.
These are also of interest in principle, though these will likely be significantly less expensive than $T$ gates in the error-corrected setting in which the $T$-count becomes a meaningful quantity to reduce.

Almost all of our results were computed using a personal laptop (Dual-core 2.5\,GHz Intel i7-6500U with 8\,GiB of RAM), with either 3 independent runs of TODD, or 10 independent runs in the case of the circuits $\mathrm{GF}(2^4)\text{-mult}$ and $\mathrm{GF}(2^5)\text{-mult}$.
For the circuits $\mathrm{GF}(2^k)\text{-mult}$ for $6 \leqslant k \leqslant 8$, we instead performed $40$ runs of TODD on Dalhousie's Mathstat Cluster (each run being performed on a separate core), taking about 5 hours in total between these circuits.
The circuits in Table~\ref{table:results} on which we demonstrate our results are those which act on 35 qubits or fewer after the stage of replacing Hadamard gates with gadgets involving auxiliary qubits.

The circuits which were obtained using our techniques may be found on GitHub [\url{https://github.com/onestruggler/stomp}].
As our main aim was to consider reductions in $T$-count, our algorithm ignores the possibility that the measurement outcomes on the auxiliary qubits could be anything but $\ket{\texttt+}$: in the event of a $\ket{\texttt-}$ outcome, additional Clifford group operations would be required, which however would not affect the $T$-count.
We verified our circuits using \texttt{feynver}~\cite{Amy-2018}, which was extended to accomodate circuits involving post-selection of $\ket{\texttt{+}}$ states on qubits which are maximally entangled with a set of other qubits.

Our results show that our techniques, simple as they are, are competitive with the best known techniques for reducing $T$ count.
In some cases, the PHAGE tactics STOMP\;4 and STOMP\;5 match or even surpass the best known results which were known.
In other cases, it is apparent that the results achievable by supplementing our techniques with TODD are better than those which were previously known with TODD and also better than only using STOMP\;4 and STOMP\;5.
Note that even when our results do not match the best known prior results, they often differ from the best known $T$-count only by $1$.

The particular PHAGE tactics which we used to obtain these results, and the way in which we deploy them, are (apart from TODD) very simple.
We expect that better results should be achievable by a more refined approach to using these concepts, within the more general framework which we have described of deploying PHAGE tactics.

\vspace*{-1ex}
\section{Discussion}

\vspace*{-1ex}
\subsection{General observations}

It seems to us that the ZX calculus not only lends itself to analysis in terms of \pp\ operations, but also leads directly to the idea of analysing $T$-count in terms of the \pp\ operations and phase gadgets.
This is particularly the case when considering circuit transformations such as those of Ref.~\cite{HC-2018} which isolate a layer of diagonal operators by commuting CNOT gates past them.

Much of our analysis clearly generalises beyond the case of reduction of $T$-count (as a measure of the complexity of a $\mathcal D_3^n$ circuit), to simplifications of $\mathcal D_k^n$ circuits.
We expect that simple generalisations of Eqn.~\eqref{gadgetdec} would provide the opportunity to explore more general simplification of diagonal circuits.

\vspace*{-1ex}
\subsection{Towards better strategies for PHAGE tactics}

Our work motivates the concept of a PHAGE tactic (simplifying a part of a circuit by selecting the best identity to apply from a family of identities), and of the importance of strategically choosing identities to apply.
The latter concept is one which is absent from our actual results, but would clearly be important to develop more efficient techniques to make use of spider nest PHAGE tactics.
As the problem of reducing $T$-count is closely related to difficult decoding or tensor-decomposition problems, it is important to find ways to divide the problem into more approachable parts: the strategy/tactic distinction is one way in which this might be done, in which the development of effective ``tactics'' which are useful in some circumstances may be the easier part, and the development of effective ``strategies'' to deploy those tactics may be the more difficult part.

We now contemplate the form that a nuanced strategy to apply spider nest PHAGE tactics could take.
A possible approach would be to compute the smallest number of ``usable'' gadgets (phase gadgets with non-trivial contribution to the $T$-count) of different sizes, which are required for some PHAGE tactic to possibly be useful, and then identify subsystems which may have the appropriate number of usable gadgets.
This motivates the problem of finding ``dense'' collections of usable $T$-phase gadgets.
Any collection of phase $m$-gadgets which are not essentially independent of one another must have some significant overlap: this motivates measuring the \emph{density of $T$-phase gadgets} at each qubit $q$ --- which we define by
\begin{equation}
\label{eqn:Tdensity}
  d(q) := \sum_{k \ge 1} \frac{\#(\text{$T$-phase $k$-gadgets which act on $q$})}{k}  \;.
\end{equation}

\noindent
We also define $d_3(q)$, the \emph{3-max density (of $T$-phase gadgets)}, which is the same sum but for $1 \le k \le 3$.
It is easy to show that $d_3(q) \le \bigl(\tfrac{1}{18} + O(1/n)\bigr) \cdot n^3$; on any qubit or collection of qubits where $d(q)$ significantly exceeds this bound, there must be several $T$-phase $m$-gadgets for $m > 3$, and it may be helpful to apply Eqn.~\eqref{gadgetdec} to decompose these into gadgets on at most $3$ qubits.
Having ensured that the circuit does not have an obvious excess of large $T$-gadgets, we may then attempt to apply a PHAGE tactic any large collections of ``usable'' gadgets that we can find on subsystems of different sizes.
This suggests a strategy along the following lines (which may be repeated several times):
\begin{algenum}
  \item
  Compute density of $T$-phase gadgets acting on each qubit (\ie,~the $k \in \{1,2,3\}$ terms of Eqn.~\eqref{eqn:Tdensity}).
  Determine the largest integer $N \ge 5$, such that the sum of the $N$ largest $3$-densities is at least $\mathbf T'_N$.
  (If no $N \ge 5$ satisfies this, then let $N = 4$.)
\item
  For each $k \in \{4,5,\ldots,N\}$, compute the \emph{score} for each qubit as the sum of the densities of those $m$-gadgets (for $1 \le m \le 3$) which are useful.
\item
  Again for each $k$, rank each qubit in order of descending score, and compute $r(k)$ to be the ``lowest'' rank such that the sum of the scores of the qubits ranked $\{1,\,{r(k)\!-\!k\!+\!2},\,{r(k)\!-\!k\!+\!3},\,\ldots,\,r(k)\}$ is at least half of the smallest $T$-count of some spider-nest identity on $k$ qubits.
  Then, let $M(k)$ be the sum of the scores of the qubits ranked from $1$ to $r(k)$, so that $M(k)$ is proportional to the average total score of a uniformly random subset of these qubits.
\item
  Repeatedly sample (a polynomial number of times) from integers $k$, with probability proportional to $M(k)$; and for each sample attempt to find a subset of size $k$ from among the qubits with the highest scores $r(k)$, in which we may reduce the $T$-count by applying a spider nest identity.
  (We may attempt to find such a subset of size $k$ by breadth-first-search on the hypergraph of $T$-gadgets).
  Compute the \emph{value} of this set as the $T$-count reduction that can be realised on this subset.
\item
  If any set with positive value was found, realise a $T$-count reduction by applying an identity to the vertex-set with the largest value.
\end{algenum}

\vspace*{-1ex}
\subsection*{Acknowledgements.}

Our techniques were realised using some functionality from Quipper~\cite{Quipper}, and our results were verified with Feynman~\cite{Amy-2018}.
We wish to extend a very special thanks to Matthew Amy, who wrote a small extension of \texttt{feynver} to allow verification of procedures which post-select the $\ket{\texttt+}$ state, for the express purpose of helping us to independently verify the correctness of our reductions.

This work was performed in partnership with Cambridge Quantum Computing under the EPSRC Impact Acceleration Award ``Compilation and cost-reduction of quantum computations via ZX-calculus''.
N.\;de\;Beaudrap is further supported by the EPSRC National Hub in Networked Quantum Information Technologies (NQIT.org), and by a Fellowship funded by a gift from Tencent Holdings (tencent.com).
X.\;Bian is supported by NSERC and by AFOSR under Award No. FA9550-15-1-0331.
Q.\;Wang is supported by the AFOSR grant FA2386-18-1-4028.
Our results were made possible in part by the use of the Dalhousie University Mathstat Cluster [\url{https://www.mathstat.dal.ca/cluster/doku.php}].

We thank Earl Campbell, Luke Heyfron, and Alexander Cowtan for helpful discussions; and Aleks Kissinger and John van de Wetering for their interest and their feedback on earlier drafts of this work.
X.\;Bian would like to thank his Ph.D. supervisor Peter Selinger for his support.

\bibliographystyle{eptcs}
\bibliography{generic}

\appendix

\newpage
\section{\Pp\ operators as generators of $\mathcal D_k^n$}
\label{apx:generators-Dkn}

We show in this Section that, together with arbitrary global phases, the operators $D_{S,k} = \exp(-\frac{i\pi}{2^k}\,Z_S)$ for $S \subseteq [n]$ generate the group $\mathcal D_k^n$. 

Let $\mathcal M_1^n = \mathcal P_n$, and for $k \ge 2$, let $\mathcal M_k^n$ consist of all products of elements of $\mathcal D_k^n$ with products of CNOT and $X$ on various qubits.
As $\mathcal D_k^n$ is preserved under conjugation by CNOT and $X$ operations, it is easy to show that $\mathcal M_k^n$ forms a group for each $k$, and that in particular that operators in $\mathcal M_k^n$ decompose as a product $U_D U_X$ for  $U_D \in \mathcal D_k^n$ and $U_X$ a circuit of CNOT and $X$ gates.
\begin{lemma}
  For $k \ge 2$ an integer, $\mathcal D_k^n \subseteq \mathcal M_k^n \subseteq \mathcal C_k^n$.
\end{lemma}
\begin{proof}
  For $k = 2$, elements of $\mathcal M_k^n$ are Clifford circuits by construction.
  For $k > 2$, consider $U = U_X U_D \in \mathcal M_k^n$, where $U_X$ is a product of CNOT and $X$ gates, and $U_D \in \mathcal D_k^n$.
  Then for any $P \in \mathcal P_n$, we have $U P U\herm = U_D U_X P U_X\herm U_D\herm = U_D Q U_D\herm$, for $Q = U_X P U_X\herm \in \mathcal P_n$.
  As $U_D Q U_D\herm \in \mathcal C_{k-1}^n$\,, the Lemma follows.  
\end{proof}
\begin{lemma}
  For integers $n,k \ge 1$ and $S \subset [n]$, $D_{\!\!\:S,k} \in \mathcal D_k^n$.
\end{lemma}
\begin{proof}
Note that $D_{\!\!\:S,1} = -i Z_{s_1} \otimes \cdots \otimes Z_{s_m} = -i Z_S \in \mathcal P_n = \mathcal C_1^n$ for any $S \subseteq [n]$.
Also, by definition we have $D_{S,k{-}1} = D_{S,k{\!\!\;+\!}1}^{\;2}$ for any $k \geqslant 2$.
By decomposing any Pauli operator $P \in \mathcal P_n$ into a product $P \propto X_{A} Z_{B}$ for sets $A,B \subseteq [n]$, it is easy to see that $D_{\!\!\:S,\;\!k} \in \mathcal D_k^n$: we have
\begin{equation}
\begin{aligned}[b]
    D_{\!\!\:S,\!\;k}^{\phantom{\mathllap{-1}}} \, P \, D_{\!\!\:S,\!\;k}^{-1}  
  \;=\;
    \exp\bigl(-\tfrac{i\pi}{2^k} Z_S\bigr) \, X_A \,Z_B \, \exp\bigl(\tfrac{i\pi}{2^k} Z_S\bigr)
  \;&=\;
    \exp\bigl(-\tfrac{i\pi}{2^k} Z_S\bigr) \,
    \exp\Bigl(\tfrac{i\pi}{2^k} X_A \,Z_S\, X_A\herm\Bigr) \, X_A \,Z_B \,
  \\&=\;
    \exp\bigl(-\tfrac{i\pi}{2^k} Z_S\bigr) \,
    \exp\Bigl((-1)^{\vec x\sur{S} \cdot\!\; \vec x\sur{A}} \tfrac{i\pi}{2^k} \,Z_S\Bigr) \, X_A \,Z_B \,
  \\&=\;
    \begin{cases}
      X_A \,Z_B,  &\text{if $\vec x\sur{S} \cdot\:\! \vec x\sur{A} = 0$},
    \\
      \exp\bigl(-\tfrac{2\pi i}{2^k} Z_S\bigr) \, X_A \, Z_B,
        & \text{if $\vec x\sur{S} \cdot\:\! \vec x\sur{A} = 1$};
    \end{cases}
\end{aligned}  
\end{equation}
in either case, $D_{\!\!\:S,\!\;k}^{\phantom{\mathllap{-1}}} \, P \, D_{\!\!\:S,\!\;k}^{-1} \in \mathcal M_{k{-}1}^n \subseteq \mathcal C_{k{-}1}^n$\,.
Then $D_{\!\!\:S,\!\;k} \in \mathcal C_k^n$, and is therefore an element of $\mathcal D_k^n$.
\end{proof}

\begin{lemma}
  For any $n,k \ge 1$, any $V \in \mathcal D_k^n$ is proportional to a product of operators $D_{\!\!\:S,\!\;k}$\,  for $S \subseteq [n]$.
\end{lemma}
\begin{proof}
Consider a decomposition of $V$ into a product of operators
$
  V \!= \prod_z V_z 
$ 
for $z$ ranging over $\{0,1\}^n$, where $\bra{z} V_z \ket{z} = \bra{z} V \ket{z} = \exp(i\theta_z)$ and where $\bra{y} V_z \ket{y} = 1$ for all $y \ne z$.
We may then express the operator $V_z$ as an exponential of a rank-1 projector on $n$ qubits:
\begin{equation}
\label{eqn:decomposeCtrlPhase}
\begin{aligned}[b]{}
\mspace{-12mu}
    V_z
  =\!\!\;
    \exp\Bigl(\! i \!\; \theta_z \bigl(\ket{z_1}\!\!\bra{z_1} \otimes \cdots \otimes \ket{z_n}\!\!\bra{z_n} \bigr) \Bigr)
  \!\!\;&=\!\;
    \exp\Bigl( \frac{i \theta_z }{2^n} \bigl[(\idop \!+\! (-1)^{z_1} Z) \!\otimes\!\!\: \cdots \!\!\:\otimes\! (\idop \!+\! (-1)^{z_n} Z) \bigr] \!\Bigr)
  \\&=\!\!\:
    \prod_{S \subset [n]} \!
      \exp\Bigl( \frac{i \theta_z }{2^n} \bigotimes_{j\in S} \,(-1)^{z_j} Z_j \!\!\:\Bigr)
  =\,
      \exp\Bigl(
				\sum_{S \subset [n]}\!\!
				\frac{i (-1)^{z \,\cdot\, \vec x\sur{S}} \!\!\:\theta_z }{2^n} \, Z_S
			\!\!\:\Bigr)
      .\!\!\!
\end{aligned}
\end{equation}%
Taking the product over $z \in \{0,1\}^n$, we then have
\begin{equation}
\begin{aligned}[b]{}
    V
  &=\!\!
    \prod_{z \in \{0,1\}^n} \!\!\!V_z
  \;=\;
     \exp\Bigl(
         \sum_{S \subset [n]}\!
					i \!\:\hat{\!\:\theta\!\!\:}_S\, Z_S
			\Bigr),
\end{aligned}
\mspace{-27mu}
\end{equation}%
where $\hat{\!\:\theta\!\!\:}_S = \sum_{\:z}\, (-1)^{z \,\cdot\, \vec x\sur{S}} \theta_z/2^n$ for the sake of brevity.
For $j \in [n]$, consider the effect of conjugation of $X_j$ by $V$: we have
\begin{small}%
\begin{equation}
  \begin{aligned}[b]{}
			V \!\!\: X_{\!\!\:j} V\herm\!
    \;&=\;
     \exp\Bigl(
         \sum_{S \subset [n]}\!\!
					i \hat{\!\:\theta\!\!\:}_S Z_S
			\Bigr)\!\!\;
     \exp\Bigl(
         \sum_{S' \subset [n]}\!\!
					i \hat{\!\:\theta\!\!\:}_{S'} X_{\!\!\:j} \!\!\; Z_{S'} \!\!\; X_{\!\!\:j}\herm
			\Bigr)
			X_j
    \\&=\;
     \exp\Bigl(
         \sum_{S \subset [n]}\!
					i \hat{\!\:\theta\!\!\:}_S \bigl[ Z_S + X_{\!\!\:j} Z_S X_{\!\!\:j}\herm \bigr]
			\Bigr)
			X_j
    \;=\;
     \exp\biggl(
         \sum_{\substack{S \subset [n] \\ j \in S}}\!
					2i \hat{\!\:\theta\!\!\:}_S Z_S 
			\biggr)
			X_{\!\!\:j} \,=:\, U_{[j]} X_{\!\!\:j}\,.
  \end{aligned}
  \mspace{-72mu}
\end{equation}%
\end{small}%
It follows that $U_{[j]} \in \mathcal D_{k{-}1}^n$, and that $U_{[j]}^{2^{k\text{--}2}}$ is a Pauli operator.
That is, the operator
\begin{equation}{}
\mspace{-36mu}
\begin{aligned}[b]
	U_{[j]}^{2^{k\text{--}2}}
  =\;
       \exp\biggl(
         \sum_{\substack{S \subset [n] \\ j \in S}}\!
					2^{k-1} i \hat{\!\:\theta\!\!\:}_S Z_S
			\biggr)
  &=\;
    \prod_{\substack{S \subseteq [n] \\ j \in S}}
      \!\Bigl(\!\!\:\cos\bigl(2^{k-1}\!\!\; \hat{\!\:\theta\!\!\:}_S\!\!\;\bigr) \!\!\; \idop \!+\! i \sin\bigl(2^{k-1}\!\!\;\hat{\!\:\theta\!\!\:}_S\!\!\;\bigr) \!\!\; Z_S \!\Bigr)
\end{aligned}
\mspace{-36mu}
\end{equation}
is a tensor product of $Z$ operations.
By the linear independence of the operators $Z_S$, it follows that every factor $\exp(2^{k-1} i \hat{\!\:\theta\!\!\:}_S Z_S)$ is either $\idop$ or $Z_S$, for $j \in S$.
As this result holds for all $j$, we obtain the same result for every non-empty set $S$.
This implies that $2^{k-1} \hat{\!\:\theta\!\!\:}_S \in \tfrac{\pi}{2}\Z$ for all $S \ne \varnothing$, or equivalently that $\hat{\!\:\theta\!\!\:}_S = m_s \pi/2^k$ for some $m \in \Z$.
It follows that $\exp(i \hat{\!\:\theta\!\!\:}_S Z_S) = D_{\!\!\:S,k}^{-m_S}$, so that $V \propto \smash{\prod_{\,S\big.} \, D_{\!\!\:S,k}^{\;-m_S}}$ for $S$ ranging over non-empty subsets of $[n]$.
\end{proof}

\newpage
\section{Proof of gadget decomposition}
\label{sec:n-gadgetDecomposition}

Here we provide a proof of Eqn.~\eqref{gadgetdec}.
We express this as a proof by induction or the proportionality (\ie,~the equality of the denotational semantics of ZX-diagrams) of a $T$-phase $n$-gadget for $n \ge 4$ on one side, and a collection of $3$-, $2$-, and $1$-gadgets as in Eqn.~\eqref{gadgetdec} on the other.

Below we use the notations $\tau:=\frac{\pi}{4}$ and $\iota:=-\frac{\pi}{4}$ for the angles of phase gadgets, written in this case inside (rather than outside) of the node to which this phase is associated.
We prove Eqn.~\eqref{gadgetdec} by induction on $n$.
The base case is the identity for $n=4$,
\begin{equation}\label{4linetgadget}  
%
          \input{./diagrams/4linetgadgetdecom.tikz}
  \;:
\end{equation}
this is the $k = 4$ case of Eqn.~\eqref{eqn:phaseParityGateIdentity}, and was shown in Ref.~\cite{ACR-2018}.
Suppose that Eqn.~\eqref{gadgetdec} holds for $n = m \geqslant 4$. Let $\sigma_m= \tfrac{1}{8}(m\!-\!2)(m\!-\!3)\pi$ and $\theta_m=-\tfrac{1}{4}(m\!-\!3)\pi$. Then for $n=m+1$, we have 
\[
%
          \input{./diagrams/4linetgadgetdecominduc.tikz}

\]
In the last diagram above, we substitute every 4-gadget with the RHS of  (\ref{4linetgadget}), and fuse together all the phase gadgets that dwell on the same lines. We assert that the resulted diagram after fusion is exactly the decomposition as presented on the RHS of (\ref{gadgetdec}) when $n=m+1$. This can be checked by calculating the phase angles of all gadget. For the 1-gadget on line 1, it comes from fusing all the 1-gadgets on line 1 which are obtained from the decomposition of all 4-gadgets connected with line. There are  ${m}\choose{2}$ such 4-gadgets, so the angle of the final 1-gadget on line 1 is
${{m}\choose{2}} \frac{\pi}{4}=\frac{(m-1)(m-2)\pi}{8}=\sigma_{m+1}$. For the final 2-gadget on line 1 and line 2, the phase angle is $\sigma_m+{{m-1}\choose{2}} \frac{-\pi}{4}=-\frac{(m-2)\pi}{4}= \theta_{m+1}$. Similarly, one can check that the 1-gadget on each line has phase angle $\sigma_{m+1}$, 2-gadget on every two lines has phase angle $\theta_{m+1}$, and 3-gadget on every three lines has phase angle $\frac{\pi}{4}$. 

Therefore, (\ref{gadgetdec}) holds for $n=m+1$. This completes the proof. 

\end{document}

%% file: tikzstyles.tex
\tikzstyle{env}=[copoint,regular polygon rotate=0,minimum width=0.2cm, fill=black]

\tikzstyle{every picture}=[baseline=-0.25em]
\tikzstyle{dotpic}=[scale=0.5]
\tikzstyle{diredges}=[every to/.style={diredge}]
\tikzstyle{dot graph}=[shorten <=-0.1mm,shorten >=-0.1mm,scale=0.6]
\tikzstyle{plot point}=[circle,fill=black,minimum width=2mm,inner sep=0]

\tikzstyle{braceedge}=[decorate,decoration={brace,amplitude=2mm,raise=-1mm}]
\tikzstyle{small braceedge}=[decorate,decoration={brace,amplitude=1mm,raise=-1mm}]
\tikzstyle{left hook arrow}=[left hook-latex]
\tikzstyle{right hook arrow}=[right hook-latex]

\tikzstyle{dtriangle}=[fill=yellow,draw=black,shape=isosceles triangle,shape border rotate=-90,isosceles triangle stretches=true,inner sep=1pt,minimum width=0.4cm,minimum height=3mm]
\tikzstyle{vtriang}=[fill=yellow,draw=black,shape=isosceles triangle,shape border rotate=180,isosceles triangle stretches=true,inner sep=1pt,minimum width=0.4cm,minimum height=3mm]
\tikzstyle{triangle}=[fill=yellow,draw=black,shape=isosceles triangle,shape border rotate=90,isosceles triangle stretches=true,inner sep=1pt,minimum width=0.4cm,minimum height=3mm]
\tikzstyle{H box}=[rectangle,fill=yellow,draw=black,xscale=1.0,yscale=1.0, inner sep=1.pt]
\tikzstyle{gbox}=[rectangle,fill=green,draw=black,xscale=1.0,yscale=1.0, inner sep=1.pt]
\tikzstyle{rbox}=[rectangle,fill=red,draw=black,xscale=1.0,yscale=1.0, inner sep=1.pt]

\tikzstyle{bn}=[circle,fill=black,draw=black,scale=.4]
\tikzstyle{wn}=[circle,fill=white,draw=black,scale=.6]
\tikzstyle{dn}=[circle,fill=none,draw=gray]

\tikzstyle{black dot}=[inner sep=0.7mm,minimum width=0pt,minimum height=0pt,fill=black,draw=black,shape=circle]

\tikzstyle{dot}=[black dot]
\tikzstyle{smalldot}=[inner sep=0.2mm,minimum width=0pt,minimum height=0pt,fill=black,draw=black,shape=circle]
\tikzstyle{white dot}=[dot,fill=white]
\tikzstyle{antipode}=[white dot,inner sep=0.3mm,font=\footnotesize]
\tikzstyle{smallwhitedot}=[smalldot,fill=white]
\tikzstyle{alt white dot}=[white dot,label={[xshift=3.07mm,yshift=-0.05mm,font=\footnotesize]left:$*$}]
\tikzstyle{gray dot}=[dot,fill=gray!40!white]
\tikzstyle{smallgraydot}=[smalldot,fill=gray!40!white]
\tikzstyle{box vertex}=[draw=black,rectangle]
\tikzstyle{small box}=[box vertex,fill=white]
\tikzstyle{whitebg}=[fill=white,inner sep=2pt]
\tikzstyle{graph state vertex}=[sg vertex,fill=black]

\tikzstyle{wide copoint}=[fill=white,draw=black,shape=isosceles triangle,shape border rotate=90,isosceles triangle stretches=true,inner sep=1pt,minimum width=1.5cm,minimum height=5mm]
\tikzstyle{wide point}=[fill=white,draw=black,shape=isosceles triangle,shape border rotate=-90,isosceles triangle stretches=true,inner sep=1pt,minimum width=1.5cm,minimum height=4mm]
\tikzstyle{very wide copoint}=[fill=white,draw=black,shape=isosceles triangle,shape border rotate=-90,isosceles triangle stretches=true,inner sep=1pt,minimum width=2.5cm,minimum height=4mm]
\tikzstyle{very wide empty copoint}=[draw=black,shape=isosceles triangle,shape border rotate=-90,isosceles triangle stretches=true,inner sep=1pt,minimum width=2.5cm,minimum height=4mm]
\tikzstyle{symm}=[ultra thick,shorten <=-1mm,shorten >=-1mm]

\tikzstyle{square box}=[rectangle,fill=white,draw=black,minimum height=5mm,minimum width=5mm,font=\small]
\tikzstyle{square gray box}=[rectangle,fill=gray!30,draw=black,minimum height=6mm,minimum width=6mm]
\tikzstyle{copoint}=[regular polygon,regular polygon sides=3,draw=black,scale=0.75,inner sep=-0.5pt,minimum width=7mm,fill=white]
\tikzstyle{point}=[regular polygon,regular polygon sides=3,draw=black,scale=0.75,inner sep=-0.5pt,minimum width=7mm,fill=white,regular polygon rotate=180]
\tikzstyle{gray point}=[point,fill=gray!40!white]
\tikzstyle{gray copoint}=[copoint,fill=gray!40!white]

\newcommand{\edgearrow}{{\arrow[black]{>}}}
\newcommand{\edgetick}{{\arrow[black,scale=0.7,very thick]{|}}}

\tikzstyle{diredge}=[->]
\tikzstyle{rdiredge}=[<-]
\tikzstyle{medium diredge}=[->]

\tikzstyle{short diredge}=[->]
\tikzstyle{halfedge}=[-)]
\tikzstyle{other halfedge}=[(-]
\tikzstyle{freeedge}=[(-)]
\tikzstyle{white edge}=[line width=5pt,white]
\tikzstyle{tick}=[postaction=decorate,decoration={markings, mark=at position 0.5 with \edgetick}]
\tikzstyle{small map edge}=[|-latex, gray!60!blue, shorten <=0.9mm, shorten >=0.5mm]
\tikzstyle{thick dashed edge}=[very thick,dashed,gray!40]
\tikzstyle{map edge}=[|-latex,very thick, gray!40, shorten <=1mm, shorten >=0.5mm]
\tikzstyle{tickedge}=[postaction=decorate,
  decoration={markings, mark=at position 0.5 with \edgetick}]
\tikzstyle{dirtickedge}=[postaction=decorate,
  decoration={markings, mark=at position 0.5 with \edgetick},
  decoration={markings, mark=at position 0.85 with \edgearrow}]
\tikzstyle{dirdoubletickedge}=[postaction=decorate,
  decoration={markings, mark=at position 0.4 with \edgetick},
  decoration={markings, mark=at position 0.6 with \edgetick},
  decoration={markings, mark=at position 0.85 with \edgearrow}]

\makeatletter
\newcommand{\boxshape}[3]{%
\pgfdeclareshape{#1}{
\inheritsavedanchors[from=rectangle] 
\inheritanchorborder[from=rectangle]
\inheritanchor[from=rectangle]{center}
\inheritanchor[from=rectangle]{north}
\inheritanchor[from=rectangle]{south}
\inheritanchor[from=rectangle]{west}
\inheritanchor[from=rectangle]{east}
\backgroundpath{
\southwest \pgf@xa=\pgf@x \pgf@ya=\pgf@y
\northeast \pgf@xb=\pgf@x \pgf@yb=\pgf@y

\@tempdima=#2
\@tempdimb=#3

\pgfpathmoveto{\pgfpoint{\pgf@xa - 5pt + \@tempdima}{\pgf@ya}}
\pgfpathlineto{\pgfpoint{\pgf@xa - 5pt - \@tempdima}{\pgf@yb}}
\pgfpathlineto{\pgfpoint{\pgf@xb + 5pt + \@tempdimb}{\pgf@yb}}
\pgfpathlineto{\pgfpoint{\pgf@xb + 5pt - \@tempdimb}{\pgf@ya}}
\pgfpathlineto{\pgfpoint{\pgf@xa - 5pt + \@tempdima}{\pgf@ya}}
\pgfpathclose
}
}}

\boxshape{NEbox}{0pt}{8pt}
\boxshape{SEbox}{0pt}{-8pt}
\boxshape{NWbox}{8pt}{0pt}
\boxshape{SWbox}{-8pt}{0pt}
\makeatother

\tikzstyle{map}=[draw,shape=NEbox,inner sep=7pt]
\tikzstyle{mapdag}=[draw,shape=SEbox,inner sep=7pt]
\tikzstyle{maptrans}=[draw,shape=SWbox,inner sep=7pt]
\tikzstyle{mapconj}=[draw,shape=NWbox,inner sep=7pt]

\tikzstyle{probs}=[shape=semicircle,fill=gray!40!white,draw=black,shape border rotate=180,minimum width=1.2cm]

\tikzstyle{arrs}=[-latex,font=\small,auto]
\tikzstyle{arrow plain}=[arrs]
\tikzstyle{arrow dashed}=[dashed,arrs]
\tikzstyle{arrow bold}=[very thick,arrs]
\tikzstyle{arrow hide}=[draw=white!0,-]
\tikzstyle{arrow reverse}=[latex-]
\tikzstyle{cdnode}=[]

\tikzstyle{gn}=[dot,fill=green,minimum width=0.3cm,inner sep=0pt]
\tikzstyle{rn}=[dot,fill=red,inner sep=0pt,minimum width=0.3cm]

\tikzstyle{rc}=[dot,thick,fill=white,draw = red,minimum width=0.3cm,inner sep=0pt]
\tikzstyle{gc}=[dot,thick,fill=white,draw= green,inner sep=0pt,minimum width=0.3cm]
\tikzstyle{bc}=[dot,thick,fill=white,draw= blue,minimum width=0.3cm]

\tikzstyle{label}=[circle,fill=white,minimum width=0.3cm]

\tikzstyle{clocklabel}=[dot,fill=yellow,draw=black,font=\tiny,inner sep=0.75pt]

\tikzstyle{rsn}=[circle split,draw,fill=red,font=\tiny,inner sep=0.75pt]
\tikzstyle{gsn}=[circle split,draw,fill=green,font=\tiny,inner sep=0.75pt]
\tikzstyle{bsn}=[circle split,draw,fill=blue,font=\tiny,inner sep=0.75pt]

\tikzstyle{rsc}=[circle split,thick,draw= red,draw,fill=white,font=\tiny,inner sep=0.75pt]
\tikzstyle{gsc}=[circle split,thick,draw= green,draw,fill=white,font=\tiny,inner sep=0.75pt]
\tikzstyle{bsc}=[circle split,thick,draw= blue,draw,fill=white,font=\tiny,inner sep=0.75pt]

\tikzstyle{cnot}=[fill=white,shape=circle,inner sep=-1.4pt]
\tikzstyle{wire label}=[font=\tiny, auto]

%% file: defs.tex
\tikzstyle{cdiag}=[matrix of math nodes, row sep=3em, column sep=3em, text height=1.5ex, text depth=0.25ex,inner sep=0.5em]
\tikzstyle{arrow above}=[transform canvas={yshift=0.5ex}]
\tikzstyle{arrow below}=[transform canvas={yshift=-0.5ex}]


%% file: zx.tikzstyles

\tikzstyle{box}=[shape=rectangle, text height=1.5ex, text depth=0.25ex, yshift=0.5mm, fill=white, draw=black, minimum height=3mm, yshift=-0.5mm, minimum width=3mm, font={\small}]
\tikzstyle{Z dot}=[fill={rgb,255:red,216; green,248; blue,216}, inner sep=0mm, minimum size=2mm, shape=circle, draw=black, inner sep=2pt]
\tikzstyle{X dot}=[Z dot, fill={rgb,255:red,227; green,145; blue,145}, shape=circle, draw=black]
\tikzstyle{blank dot}=[Z dot, fill={rgb,255:red,255; green,255; blue,255}, shape=circle, draw=black]
\tikzstyle{Z phase dot}=[fill={rgb,255:red,216; green,248; blue,216},  minimum size=5mm, font={\footnotesize\boldmath\boldsymb}, shape=circle, rounded corners=2mm, inner sep=0.2mm, outer sep=-2mm, scale=0.8, tikzit shape=circle, draw=black]
\tikzstyle{X phase dot}=[Z phase dot, fill={rgb,255:red,227; green,145; blue,145}, font={\footnotesize\boldmath\boldsymb}, inner sep=2pt]
\tikzstyle{hadamard}=[fill=yellow!50!white, draw=black, shape=rectangle, inner sep=0.6mm, minimum height=1.5mm, minimum width=1.5mm]
\tikzstyle{vertex}=[inner sep=0mm, minimum size=1.5mm, shape=circle, draw=black, fill=black]
\tikzstyle{sq}=[inner sep=0mm, minimum size=1.5mm, shape=circle, draw=black, fill=white]
\tikzstyle{vertex set}=[inner sep=0mm, minimum size=1.5mm, shape=circle, draw=black, fill=white, font={\footnotesize\boldmath}]

\tikzstyle{hadamard edge}=[-, color=blue, dashed, dash pattern=on 2pt off 0.7pt]
\tikzstyle{brace edge}=[-, tikzit draw=blue, decorate, decoration={brace,amplitude=1mm,raise=-1mm}]

%% file: diagrams/commute-x-gadget.tikz
  \mspace{-18mu}
  \begin{minipage}{23mm}
    \begin{tikzpicture}
      \def\dx{0.4}
      \def\dy{0.5}
      \coordinate (0) at (0,0);
      \coordinate (x1-0) at (0);
      \coordinate (x2-0) at ($(x1-0) + (0,\dy)$);
      \coordinate (x3-0) at ($(x2-0) + (0,\dy)$);
      \coordinate (x4-0) at ($(x3-0) + (0,\dy)$);
      \xdef\u{0}
      \foreach \t in {1,...,5} {%
        \foreach \k in {1,...,4} {%
          \coordinate (x\k-\t) at ($(x\k-\u) + (\dx,0)$);
          \ifnum\k=3%
            \ifnum\t=3
              \coordinate (x3-3) at ($(x3-3) + ({\dx/2},0)$);
            \fi
          \else
            \draw (x\k-\u) -- (x\k-\t);
          \fi
        }
        \xdef\u{\t}
      }
      \node at ($(x3-2) + (0,0.1)$) {$\vdots$};
      \foreach \k in {1,2,4} {%
        \draw (x\k-2) -- (x3-3);
        \filldraw [Z dot] (x\k-2) circle (3pt);
      }
      \draw (x3-3) -- (x3-4);
      \filldraw [X dot] (x3-3) circle (3pt);
      \filldraw [Z dot] (x3-4) circle (3pt) node [anchor=180,font=\footnotesize] {\;$\theta$}; 
      \filldraw [X dot] (x1-1) circle (3pt)
        node [anchor=north,font=\scriptsize] {$\phantom\vert\pi\phantom\vert$};
    \end{tikzpicture}
  \end{minipage}
  \;\;\longrightarrow\;\;
  \begin{minipage}{26mm}
    \hspace*{0.5mm}
    \begin{tikzpicture}
      \def\dx{0.4}
      \def\dy{0.5}
      \coordinate (0) at (0,0);
      \coordinate (x1-0) at (0);
      \coordinate (x2-0) at ($(x1-0) + (0,\dy)$);
      \coordinate (x3-0) at ($(x2-0) + (0,\dy)$);
      \coordinate (x4-0) at ($(x3-0) + (0,\dy)$);
      \xdef\u{0}
      \foreach \t in {1,...,5} {%
        \foreach \k in {1,...,4} {%
          \coordinate (x\k-\t) at ($(x\k-\u) + (\dx,0)$);
          \ifnum\k=3%
          \else
            \draw (x\k-\u) -- (x\k-\t);
          \fi
        }
        \xdef\u{\t}
      }
      \node at ($(x3-1) + (0,0.1)$) {$\vdots$};
      \foreach \k in {1,2,4} {%
        \draw (x\k-1) -- (x3-2);
        \filldraw [Z dot] (x\k-1) circle (3pt);
      }
      \draw (x3-2) -- (x3-3);
      \filldraw [X dot] (x3-2) circle (3pt);
      \filldraw [Z dot] (x3-3) circle (3pt)
        node [anchor=west,font=\scriptsize] {\;$-\theta$};
      \filldraw [X dot] (x1-3) circle (3pt)
        node [anchor=north,font=\scriptsize] {$\phantom\vert\pi\phantom\vert$};
    \end{tikzpicture}
  \end{minipage}

%% file: diagrams/commute-cond-x-gadget.tikz
  \begin{minipage}{23mm}
    \begin{tikzpicture}
      \def\dx{0.4}
      \def\dy{0.5}
      \coordinate (0) at (0,0);
      \coordinate (x1-0) at (0);
      \coordinate (x2-0) at ($(x1-0) + (0,\dy)$);
      \coordinate (x3-0) at ($(x2-0) + (0,\dy)$);
      \coordinate (x4-0) at ($(x3-0) + (0,\dy)$);
      \xdef\u{0}
      \foreach \t in {1,...,5} {%
        \foreach \k in {1,...,4} {%
          \coordinate (x\k-\t) at ($(x\k-\u) + (\dx,0)$);
          \ifnum\k=3%
            \ifnum\t=3
              \coordinate (x3-3) at ($(x3-3) + ({\dx/2},0)$);
            \fi
          \else
            \draw (x\k-\u) -- (x\k-\t);
          \fi
        }
        \xdef\u{\t}
      }
      \node at ($(x3-2) + (0,0.1)$) {$\vdots$};
      \foreach \k in {1,2,4} {%
        \draw (x\k-2) -- (x3-3);
        \filldraw [Z dot] (x\k-2) circle (3pt);
      }
      \draw (x3-3) -- (x3-4);
      \filldraw [X dot] (x3-3) circle (3pt);
      \filldraw [Z dot] (x3-4) circle (3pt) node [anchor=180,font=\footnotesize] {\;$\theta$}; 
      \filldraw [X dot] (x1-1) circle (3pt)
        node [anchor=north,font=\scriptsize,inner sep=0pt] {$\pi\!\!\:,\!\{\!\!\;s\!\!\;\}$};
    \end{tikzpicture}
  \end{minipage}
  \longrightarrow{}
    \begin{minipage}{40mm}
    \hspace*{0.5mm}
    \begin{tikzpicture}
      \def\dx{0.4}
      \def\dy{0.5}
      \coordinate (0) at (0,0);
      \coordinate (x1-0) at (0);
      \coordinate (x2-0) at ($(x1-0) + (0,\dy)$);
      \coordinate (x3-0) at ($(x2-0) + (0,\dy)$);
      \coordinate (x4-0) at ($(x3-0) + (0,\dy)$);
      \xdef\u{0}
      \foreach \t in {1,...,7} {%
        \foreach \k in {1,...,4} {%
					\coordinate (x\k-\t) at ($(x\k-\u) + (\dx,0)$);
          \ifnum\k=3%
          \else
            \draw (x\k-\u) -- (x\k-\t);
          \fi
        }
        \xdef\u{\t}
      }
      \node at ($(x3-1) + (0,0.1)$) {$\vdots$};
      \node at ($(x3-4) + (0,0.1)$) {$\vdots$};
      \coordinate (x3-2') at ($(x3-2) + (0,{\dy/2})$);  
      \coordinate (x3-2'') at ($(x3-5) + (0,{-\dy/2})$);
      \coordinate (x3-3') at ($(x3-2') + (\dx,0)$);  
      \coordinate (x3-3'') at ($(x3-2'') + (\dx,0)$);
      \draw (x3-2') -- (x3-3');
      \draw (x3-2'') -- (x3-3'');
			\foreach \t/\p in {1/',4/''} {%
				\foreach \u in {1,2,4} {%
					\draw (x\u-\t) -- (x3-2\p);
					\filldraw [Z dot] (x\u-\t) circle (3pt);
				}
			}
      \filldraw [Z dot] (x3-2') circle (3pt);
      \filldraw [Z dot] (x3-2'') circle (3pt);
      \filldraw [X dot] (x3-2') circle (3pt);
      \filldraw [X dot] (x3-2'') circle (3pt);      
      \filldraw [Z dot] (x3-3') circle (3pt) node [anchor=north,inner sep=1pt,font=\footnotesize] {$\phantom\vert\theta\phantom\vert$}; 
      \filldraw [Z dot] (x3-3'') circle (3pt) node [anchor=south,inner sep=-11pt,font=\footnotesize] {\;\;$-\!2\theta\!\!\:,\!\{\!\!\;s\!\!\;\}$}; 
      \filldraw [X dot] (x1-6) circle (3pt)
        node [anchor=north,font=\scriptsize] {$\pi\!\!\:,\!\{\!\!\;s\!\!\;\}$};
    \end{tikzpicture}
  \end{minipage}
  \mspace{-48mu}

%% file: diagrams/commute-cx-gadget-a.tikz
  \begin{minipage}{23mm}
    \begin{tikzpicture}
      \def\dx{0.4}
      \def\dy{0.5}
      \coordinate (0) at (0,0);
      \coordinate (x1-0) at (0);
      \coordinate (x2-0) at ($(x1-0) + (0,\dy)$);
      \coordinate (x3-0) at ($(x2-0) + (0,\dy)$);
      \coordinate (x4-0) at ($(x3-0) + (0,\dy)$);
      \xdef\u{0}
      \foreach \t in {1,...,5} {%
        \foreach \k in {1,...,4} {%
          \coordinate (x\k-\t) at ($(x\k-\u) + (\dx,0)$);
          \ifnum\k=3%
            \ifnum\t=3
              \coordinate (x3-3) at ($(x3-2) + ({2*\dx/3},0)$);
            \fi
          \else
            \draw (x\k-\u) -- (x\k-\t);
          \fi
        }
        \xdef\u{\t}
      }
      \node at ($(x3-2) + (0,0.1)$) {$\vdots$};
      \foreach \k in {2,4} {%
        \draw (x\k-2) -- (x3-3);
        \filldraw [Z dot] (x\k-2) circle (3pt);
      }
      \draw (x3-3) -- (x3-4);
      \filldraw [X dot] (x3-3) circle (3pt);
      \filldraw [Z dot] (x3-4) circle (3pt) node [anchor=180,font=\footnotesize] {\;$\theta$}; 
	  \draw (x1-1) -- (x2-1);
      \filldraw [Z dot] (x2-1) circle (3pt);
      \filldraw [X dot] (x1-1) circle (3pt);
    \end{tikzpicture}
  \end{minipage}
  \!\!\!\!\!\longrightarrow
  \begin{minipage}{23mm}
    \begin{tikzpicture}
      \def\dx{0.4}
      \def\dy{0.5}
      \coordinate (0) at (0,0);
      \coordinate (x1-0) at (0);
      \coordinate (x2-0) at ($(x1-0) + (0,\dy)$);
      \coordinate (x3-0) at ($(x2-0) + (0,\dy)$);
      \coordinate (x4-0) at ($(x3-0) + (0,\dy)$);
      \xdef\u{0}
      \foreach \t in {1,...,4} {%
        \foreach \k in {1,...,4} {%
          \coordinate (x\k-\t) at ($(x\k-\u) + (\dx,0)$);
          \ifnum\k=3%
            \ifnum\t=2
              \coordinate (x3-3) at ($(x3-2) + ({2*\dx/3},0)$);
            \fi
          \else
            \draw (x\k-\u) -- (x\k-\t);
          \fi
        }
        \xdef\u{\t}
      }
      \node at ($(x3-1) + (0,0.1)$) {$\vdots$};
      \foreach \k in {2,4} {%
        \draw (x\k-1) -- (x3-2);
        \filldraw [Z dot] (x\k-1) circle (3pt);
      }
      \draw (x3-2) -- (x3-3);
      \filldraw [X dot] (x3-2) circle (3pt);
      \filldraw [Z dot] (x3-3) circle (3pt) node [anchor=180,font=\footnotesize] {\;$\theta$}; 
	  \draw (x1-3) -- (x2-3);
      \filldraw [Z dot] (x2-3) circle (3pt);
      \filldraw [X dot] (x1-3) circle (3pt);
    \end{tikzpicture}
  \end{minipage}

%% file: diagrams/commute-cx-gadget-b.tikz
  \begin{minipage}{23mm}
    \begin{tikzpicture}
      \def\dx{0.4}
      \def\dy{0.5}
      \coordinate (0) at (0,0);
      \coordinate (x1-0) at (0);
      \coordinate (x2-0) at ($(x1-0) + (0,\dy)$);
      \coordinate (x3-0) at ($(x2-0) + (0,\dy)$);
      \coordinate (x4-0) at ($(x3-0) + (0,\dy)$);
      \xdef\u{0}
      \foreach \t in {1,...,5} {%
        \foreach \k in {1,...,4} {%
          \coordinate (x\k-\t) at ($(x\k-\u) + (\dx,0)$);
          \ifnum\k=3%
            \ifnum\t=3
              \coordinate (x3-3) at ($(x3-2) + ({2*\dx/3},0)$);
            \fi
          \else
            \draw (x\k-\u) -- (x\k-\t);
          \fi
        }
        \xdef\u{\t}
      }
      \node at ($(x3-2) + (0,0.1)$) {$\vdots$};
      \foreach \k in {2,4} {%
        \draw (x\k-2) -- (x3-3);
        \filldraw [Z dot] (x\k-2) circle (3pt);
      }
      \draw (x3-3) -- (x3-4);
      \filldraw [X dot] (x3-3) circle (3pt);
      \filldraw [Z dot] (x3-4) circle (3pt) node [anchor=180,font=\footnotesize] {\;$\theta$}; 
	  \draw (x1-1) -- (x2-1);
      \filldraw [Z dot] (x1-1) circle (3pt);
      \filldraw [X dot] (x2-1) circle (3pt);
    \end{tikzpicture}
  \end{minipage}
  \!\!\!\!\!\longrightarrow\;
  \begin{minipage}{23mm}
    \begin{tikzpicture}
      \def\dx{0.4}
      \def\dy{0.5}
      \coordinate (0) at (0,0);
      \coordinate (x1-0) at (0);
      \coordinate (x2-0) at ($(x1-0) + (0,\dy)$);
      \coordinate (x3-0) at ($(x2-0) + (0,\dy)$);
      \coordinate (x4-0) at ($(x3-0) + (0,\dy)$);
      \xdef\u{0}
      \foreach \t in {1,...,4} {%
        \foreach \k in {1,...,4} {%
          \coordinate (x\k-\t) at ($(x\k-\u) + (\dx,0)$);
          \ifnum\k=3%
            \ifnum\t=2
              \coordinate (x3-3) at ($(x3-2) + ({2*\dx/3},0)$);
            \fi
          \else
            \draw (x\k-\u) -- (x\k-\t);
          \fi
        }
        \xdef\u{\t}
      }
      \node at ($(x3-1) + (0,0.1)$) {$\vdots$};
      \foreach \k in {1,2,4} {%
        \draw (x\k-1) -- (x3-2);
        \filldraw [Z dot] (x\k-1) circle (3pt);
      }
      \draw (x3-2) -- (x3-3);
      \filldraw [X dot] (x3-2) circle (3pt);
      \filldraw [Z dot] (x3-3) circle (3pt) node [anchor=180,font=\footnotesize] {\;$\theta$}; 
	  \draw (x1-3) -- (x2-3);
      \filldraw [Z dot] (x1-3) circle (3pt);
      \filldraw [X dot] (x2-3) circle (3pt);
    \end{tikzpicture}
  \end{minipage}
  \mspace{-18mu};\qquad
  \begin{minipage}{23mm}
    \begin{tikzpicture}
      \def\dx{0.4}
      \def\dy{0.5}
      \coordinate (0) at (0,0);
      \coordinate (x1-0) at (0);
      \coordinate (x2-0) at ($(x1-0) + (0,\dy)$);
      \coordinate (x3-0) at ($(x2-0) + (0,\dy)$);
      \coordinate (x4-0) at ($(x3-0) + (0,\dy)$);
      \xdef\u{0}
      \foreach \t in {1,...,5} {%
        \foreach \k in {1,...,4} {%
          \coordinate (x\k-\t) at ($(x\k-\u) + (\dx,0)$);
          \ifnum\k=3%
            \ifnum\t=3
              \coordinate (x3-3) at ($(x3-2) + ({2*\dx/3},0)$);
            \fi
          \else
            \draw (x\k-\u) -- (x\k-\t);
          \fi
        }
        \xdef\u{\t}
      }
      \node at ($(x3-2) + (0,0.1)$) {$\vdots$};
      \foreach \k in {1,2,4} {%
        \draw (x\k-2) -- (x3-3);
        \filldraw [Z dot] (x\k-2) circle (3pt);
      }
      \draw (x3-3) -- (x3-4);
      \filldraw [X dot] (x3-3) circle (3pt);
      \filldraw [Z dot] (x3-4) circle (3pt) node [anchor=180,font=\footnotesize] {\;$\theta$}; 
	  \draw (x1-1) -- (x2-1);
      \filldraw [Z dot] (x1-1) circle (3pt);
      \filldraw [X dot] (x2-1) circle (3pt);
    \end{tikzpicture}
  \end{minipage}
  \!\!\!\!\!\longrightarrow\;
  \begin{minipage}{23mm}
    \begin{tikzpicture}
      \def\dx{0.4}
      \def\dy{0.5}
      \coordinate (0) at (0,0);
      \coordinate (x1-0) at (0);
      \coordinate (x2-0) at ($(x1-0) + (0,\dy)$);
      \coordinate (x3-0) at ($(x2-0) + (0,\dy)$);
      \coordinate (x4-0) at ($(x3-0) + (0,\dy)$);
      \xdef\u{0}
      \foreach \t in {1,...,4} {%
        \foreach \k in {1,...,4} {%
          \coordinate (x\k-\t) at ($(x\k-\u) + (\dx,0)$);
          \ifnum\k=3%
            \ifnum\t=2
              \coordinate (x3-3) at ($(x3-2) + ({2*\dx/3},0)$);
            \fi
          \else
            \draw (x\k-\u) -- (x\k-\t);
          \fi
        }
        \xdef\u{\t}
      }
      \node at ($(x3-1) + (0,0.1)$) {$\vdots$};
      \foreach \k in {2,4} {%
        \draw (x\k-1) -- (x3-2);
        \filldraw [Z dot] (x\k-1) circle (3pt);
      }
      \draw (x3-2) -- (x3-3);
      \filldraw [X dot] (x3-2) circle (3pt);
      \filldraw [Z dot] (x3-3) circle (3pt) node [anchor=180,font=\footnotesize] {\;$\theta$}; 
	  \draw (x1-3) -- (x2-3);
      \filldraw [Z dot] (x1-3) circle (3pt);
      \filldraw [X dot] (x2-3) circle (3pt);
    \end{tikzpicture}
  \end{minipage}

%% file: diagrams/gadget-fuse.tikz
    \begin{minipage}{40mm}
    \hspace*{0.5mm}
    \begin{tikzpicture}
      \def\dx{0.4}
      \def\dy{0.5}
      \coordinate (0) at (0,0);
      \coordinate (x1-0) at (0);
      \coordinate (x2-0) at ($(x1-0) + (0,\dy)$);
      \coordinate (x3-0) at ($(x2-0) + (0,\dy)$);
      \coordinate (x4-0) at ($(x3-0) + (0,\dy)$);
      \xdef\u{0}
      \foreach \t in {1,...,7} {%
        \foreach \k in {1,...,4} {%
	\coordinate (x\k-\t) at ($(x\k-\u) + (\dx,0)$);
          \ifnum\k=3%
          \else
            \draw (x\k-\u) -- (x\k-\t);
          \fi
        }
        \xdef\u{\t}
      }
      \node at ($(x3-1) + (0,0.1)$) {$\vdots$};
      \node at ($(x3-4) + (0,0.1)$) {$\vdots$};
      \coordinate (x3-2') at ($(x3-2) + (0,{-\dy/2})$);  
      \coordinate (x3-2'') at ($(x3-5) + (0,{-\dy/2})$);
      \coordinate (x3-3') at ($(x3-2') + (\dx,0)$);  
      \coordinate (x3-3'') at ($(x3-2'') + (\dx,0)$);
      \draw (x3-2') -- (x3-3');
      \draw (x3-2'') -- (x3-3'');
			\foreach \t/\p in {1/',4/''} {%
				\foreach \u in {1,2,4} {%
					\draw (x\u-\t) -- (x3-2\p);
					\filldraw [Z dot] (x\u-\t) circle (3pt);
				}
			}
      \filldraw [Z dot] (x3-2') circle (3pt);
      \filldraw [Z dot] (x3-2'') circle (3pt);
      \filldraw [X dot] (x3-2') circle (3pt);
      \filldraw [X dot] (x3-2'') circle (3pt);      
      \filldraw [Z dot] (x3-3') circle (3pt) node [anchor=south,inner sep=1pt,font=\footnotesize] {$\phantom\vert\alpha\phantom\vert$}; 
      \filldraw [Z dot] (x3-3'') circle (3pt) node [anchor=south,inner sep=-11pt,font=\footnotesize] {\;\;$\beta$}; 
    \end{tikzpicture}
  \end{minipage}
  \mspace{-36mu}\longrightarrow{}
    \begin{minipage}{40mm}
    \hspace*{0.5mm}
    \begin{tikzpicture}
      \def\dx{0.4}
      \def\dy{0.5}
      \coordinate (0) at (0,0);
      \coordinate (x1-0) at (0);
      \coordinate (x2-0) at ($(x1-0) + (0,\dy)$);
      \coordinate (x3-0) at ($(x2-0) + (0,\dy)$);
      \coordinate (x4-0) at ($(x3-0) + (0,\dy)$);
      \xdef\u{0}
      \foreach \t in {1,...,4} {%
        \foreach \k in {1,...,4} {%
					\coordinate (x\k-\t) at ($(x\k-\u) + (\dx,0)$);
          \ifnum\k=3%
          \else
            \draw (x\k-\u) -- (x\k-\t);
          \fi
        }
        \xdef\u{\t}
      }
      \node at ($(x3-1) + (0,0.1)$) {$\vdots$};
      \coordinate (x3-2') at ($(x3-2) + (0,{\dy/2})$);  
      \coordinate (x3-2'') at ($(x3-2) + (0,{-\dy/2})$);
      \coordinate (x3-3') at ($(x3-2') + (\dx,0)$);  
      \coordinate (x3-3'') at ($(x3-2'') + (\dx,0)$);
      \draw (x3-2'') -- (x3-3'');
			\foreach \t/\p in {1/''} {%
				\foreach \u in {1,2,4} {%
					\draw (x\u-\t) -- (x3-2\p);
					\filldraw [Z dot] (x\u-\t) circle (3pt);
				}
			}
      \filldraw [Z dot] (x3-2'') circle (3pt);
      \filldraw [X dot] (x3-2'') circle (3pt);      
      \filldraw [Z dot] (x3-3'') circle (3pt) node [anchor=south,inner sep=-11pt,font=\footnotesize] {$\alpha{+}\beta$}; 
    \end{tikzpicture}
  \end{minipage}
  \mspace{-72mu}

%% file: diagrams/tphasegadgetdecom.tikz
    n \;\left\{\;
  \begin{aligned}
  ~\\[-6ex]
  \begin{tikzpicture}[scale=1.125]
    \def\dx{0.25}
    \def\dy{0.875}
    \coordinate (x0-0) at (0,0);
    \coordinate (xn-0) at ($(x0-0) + (0,\dy)$);
    \coordinate (x2n-0) at ($(xn-0) + (0,\dy)$);
    \xdef\u{0}
    \foreach \t in {1,...,38} {%
      \foreach \v in {x0,xn,x2n} {%
        \coordinate (\v-\t) at ($(\v-\u) + (\dx,0)$);
        \draw (\v-\u) -- (\v-\t);
      }
      \xdef\u{\t}
    }
    \node at ($(x0-1)!0.625!(xn-1)$) {$\vdots$};
    \node at ($(xn-1)!0.625!(x2n-1)$) {$\vdots$};
    \foreach \v in {x0,xn,x2n} {%
      \node [style=Z dot] (\v-5) at (\v-5) {};
      \node at (\v-5) [anchor=270] {$\scriptstyle (n{-}2)(n{-}3)\tfrac{\pi}{8}$};
    }
    \node [style=Z dot] (x0-9) at (x0-9) {};
    \node [style=Z dot] (xn-9) at (xn-9) {};
    \node [style=X dot] (p) at ($(x0-10)!0.5!(xn-10)$) {};
    \draw (x0-9) -- (p);
    \draw (xn-9) -- (p);
    \node [style=Z dot] (ph) at ($(x0-11)!0.5!(xn-12)$) {};
    \node at (ph) [anchor=180] {$\scriptstyle \:\:-(n{-}3)\tfrac{\pi}{4}$};
    \draw (p) -- (ph);
    \node [style=Z dot] (xn-10) at (xn-10) {};
    \node [style=Z dot] (x2n-10) at (x2n-10) {};
    \node [style=X dot] (p) at ($(xn-11)!0.5!(x2n-11)$) {};
    \draw (xn-10) -- (p);
    \draw (x2n-10) -- (p);
    \node [style=Z dot] (ph) at ($(xn-12)!0.5!(x2n-13)$) {};
    \node at (ph) [anchor=180] {$\scriptstyle \:\:-(n{-}3)\tfrac{\pi}{4}$};
    \draw (p) -- (ph);
    \node [style=Z dot] (x0-18) at (x0-18) {};
    \node [style=Z dot] (x2n-18) at (x2n-18) {};
    \node [style=X dot] (p) at ($(x0-19)!0.666!(x2n-19)$) {};
    \draw (x0-18) -- (p);
    \draw (x2n-18) -- (p);
    \node [style=Z dot] (ph) at ($(x0-20)!0.666!(x2n-21)$) {};
    \node at (ph) [anchor=180] {$\scriptstyle \:\:-(n{-}3)\tfrac{\pi}{4}$};
    \draw (p) -- (ph);
    \node [style=Z dot] (x0-27) at (x0-27) {};
    \node [style=Z dot] (xn-27) at (xn-27) {};
    \node [style=Z dot] (x2n-27) at (x2n-27) {};
    \node [style=X dot] (p) at ($(x0-28)!0.333!(x2n-29)$) {};
    \draw (x0-27) -- (p);
    \draw (xn-27) -- (p);
    \draw (x2n-27) -- (p);
    \node [style=Z dot] (ph) at ($(x0-29)!0.333!(x2n-31)$) {};
    \node at (ph) [anchor=180] {$\scriptstyle \,\,\tfrac{\pi}{4}$};
    \draw (p) -- (ph);
%
%
    \node (p) [style=X dot] at ($(x0-35)!0.3125!(x2n-35)$) {};
    \node at ($(x0-34)!0.625!(xn-34)$) {$\vdots$};
    \node at ($(xn-34)!0.5875!(x2n-34)$) {$\vdots$};
    \foreach \v in {x0,xn,x2n} {%
      \node [style=Z dot] (\v-34) at (\v-34) {};
      \draw (\v-34) -- (p);
    }
    \node (ph) [style=Z dot] at ($(p) + (\dx,0)$) {};
    \node at (ph) [anchor=180] {$\scriptstyle \:-\!\!\:\tfrac{\pi}{4}$};
    \draw (p) -- (ph);
    \end{tikzpicture}
  \\[-3.5ex]~
  \end{aligned}
  \right.  
  \;\;\;
  \propto
  \;\;
\idop^{\otimes n}

%% file: diagrams/order2-spidernest.tikz
  \begin{aligned}
  ~\\[-6ex]
  \begin{tikzpicture}[scale=1]
    \def\dx{0.25}
    \def\dy{0.875}
    \coordinate (x0-1) at (0,0);
    \coordinate (xn-1) at ($(x0-0) + (0,\dy)$);
    \coordinate (x2n-1) at ($(xn-0) + (0,\dy)$);
    \coordinate (X-1) at ($(x2n-0) + (0,\dy)$);
    \xdef\u{1}
    \foreach \t in {1,...,56} {%
      \foreach \v in {X,x0,xn,x2n} {%
        \coordinate (\v-\t) at ($(\v-\u) + (\dx,0)$);
        \draw (\v-\u) -- (\v-\t);
      }
      \xdef\u{\t}
    }
    \node at ($(x0-2)!0.625!(xn-1)$) {$\vdots$};
    \node at ($(xn-2)!0.625!(x2n-1)$) {$\vdots$};
    \foreach \v in {x0,xn,x2n} {%
      \node [style=Z dot] (\v-5) at (\v-5) {};
      \node at (\v-5) [anchor=90] {$\scriptstyle (n{-}3)\tfrac{\pi}{4}$};
    }
    \node [style=Z dot] (X-5) at (X-15) {};
    \node [anchor=315] at ($(x0-3)!0.5!(x0-3)$) {$n\!\:{-}\!\!\;1 \,\left\{\;\begin{array}{c}\\[11ex] \end{array}\right.$};
    \node at (X-5) [anchor=270] {$\scriptstyle (n{-}2)(n{-}3)\tfrac{\pi}{8}$};
    \node [style=Z dot] (x0-8) at (x0-8) {};
    \node [style=Z dot] (xn-8) at (xn-8) {};
    \node [style=X dot] (p) at ($(x0-9)!0.5!(xn-9)$) {};
    \draw (x0-8) -- (p);
    \draw (xn-8) -- (p);
    \node [style=Z dot] (ph) at ($(x0-10)!0.5!(xn-11)$) {};
    \node at (ph) [anchor=180] {$\scriptstyle \:\:-\!\!\;\tfrac{\pi}{4}$};
    \draw (p) -- (ph);
    \node [style=Z dot] (xn-9) at (xn-9) {};
    \node [style=Z dot] (x2n-9) at (x2n-9) {};
    \node [style=X dot] (p) at ($(xn-10)!0.5!(x2n-10)$) {};
    \draw (xn-9) -- (p);
    \draw (x2n-9) -- (p);
    \node [style=Z dot] (ph) at ($(xn-11)!0.5!(x2n-12)$) {};
    \node at (ph) [anchor=180] {$\scriptstyle \:\:-\!\!\;\tfrac{\pi}{4}$};
    \draw (p) -- (ph);
    \node [style=Z dot] (x0-15) at (x0-15) {};
    \node [style=Z dot] (x2n-15) at (x2n-15) {};
    \node [style=X dot] (p) at ($(x0-16)!0.666!(x2n-16)$) {};
    \draw (x0-15) -- (p);
    \draw (x2n-15) -- (p);
    \node [style=Z dot] (ph) at ($(x0-17)!0.666!(x2n-18)$) {};
    \node at (ph) [anchor=180] {$\scriptstyle \:\:-\!\!\;\tfrac{\pi}{4}$};
    \draw (p) -- (ph);
    \node [style=Z dot] (x0-20) at (x0-20) {};
    \node [style=Z dot] (X-20) at (X-20) {};
    \node [style=X dot] (p) at ($(x0-21)!0.1666!(X-21)$) {};
    \draw (x0-20) -- (p);
    \draw (X-20) -- (p);
    \node [style=Z dot] (ph) at ($(x0-22)!0.1666!(X-25)$) {};
    \node at (ph) [anchor=180] {$\scriptstyle \:\:-(n{-}3)\tfrac{\pi}{4}$};
    \draw (p) -- (ph);
    \node [style=Z dot] (xn-22) at (xn-22) {};
    \node [style=Z dot] (X-22) at (X-22) {};
    \node [style=X dot] (p) at ($(xn-23)!0.25!(X-23)$) {};
    \draw (xn-22) -- (p);
    \draw (X-22) -- (p);
    \node [style=Z dot] (ph) at ($(xn-24)!0.25!(X-26)$) {};
    \node at (ph) [anchor=180] {$\scriptstyle \:\:-(n{-}3)\tfrac{\pi}{4}$};
    \draw (p) -- (ph);
    \node [style=Z dot] (x2n-24) at (x2n-24) {};
    \node [style=Z dot] (X-24) at (X-24) {};
    \node [style=X dot] (p) at ($(x2n-25)!0.5!(X-26)$) {};
    \draw (x2n-24) -- (p);
    \draw (X-24) -- (p);
    \node [style=Z dot] (ph) at ($(x2n-27)!0.5!(X-27)$) {};
    \node at (ph) [anchor=180] {$\scriptstyle \:\:-(n{-}3)\tfrac{\pi}{4}$};
    \draw (p) -- (ph);
    \node [style=Z dot] (X-33) at (X-33) {};
    \node [style=Z dot] (x0-33) at (x0-33) {};
    \node [style=Z dot] (xn-33) at (xn-33) {};
    \node [style=X dot] (p) at ($(x0-34)!0.5!(xn-35)$) {};
    \draw (X-33) -- (p);
    \draw (x0-33) -- (p);
    \draw (xn-33) -- (p);
    \node [style=Z dot] (ph) at ($(x0-36)!0.5!(xn-36)$) {};
    \node at (ph) [anchor=180] {$\scriptstyle \;\tfrac{\pi}{4}$};
    \draw (p) -- (ph);
    \node [style=Z dot] (X-35) at (X-35) {};
    \node [style=Z dot] (xn-35) at (xn-35) {};
    \node [style=Z dot] (x2n-35) at (x2n-35) {};
    \node [style=X dot] (p) at ($(xn-36)!0.5!(x2n-37)$) {};
    \draw (X-35) -- (p);
    \draw (xn-35) -- (p);
    \draw (x2n-35) -- (p);
    \node [style=Z dot] (ph) at ($(xn-38)!0.5!(x2n-38)$) {};
    \node at (ph) [anchor=180] {$\scriptstyle \;\tfrac{\pi}{4}$};
    \draw (p) -- (ph);
    \node [style=Z dot] (X-40) at (X-40) {};
    \node [style=Z dot] (x0-40) at (x0-40) {};
    \node [style=Z dot] (x2n-40) at (x2n-40) {};
    \node [style=X dot] (p) at ($(x0-41)!0.75!(x2n-42)$) {};
    \draw (X-40) -- (p);
    \draw (x0-40) -- (p);
    \draw (x2n-40) -- (p);
    \node [style=Z dot] (ph) at ($(x0-42)!0.75!(x2n-44)$) {};
    \node at (ph) [anchor=180] {$\scriptstyle \;\tfrac{\pi}{4}$};
    \draw (p) -- (ph);
    \node (p) [style=X dot] at ($(x0-47)!0.5!(xn-48)$) {};
    \node at ($(x0-46)!0.625!(xn-46)$) {$\vdots$};
    \node at ($(xn-46)!0.5875!(x2n-46)$) {$\vdots$};
    \foreach \v in {x0,xn,x2n} {%
      \node [style=Z dot] (\v-46) at (\v-46) {};
      \draw (\v-46) -- (p);
    }
    \node (ph) [style=Z dot] at ($(p) + (1.5*\dx,0)$) {};
    \node at (ph) [anchor=180] {$\scriptstyle \:-\!\!\:\tfrac{\pi}{4}$};
    \draw (p) -- (ph);
    \node (p) [style=X dot] at ($(x2n-54)!0.325!(xn-55)$) {};
    \node at ($(x0-52)!0.625!(xn-52)$) {$\vdots$};
    \node at ($(xn-52)!0.5875!(x2n-52)$) {$\vdots$};
    \foreach \v in {X,x0,xn,x2n} {%
      \node [style=Z dot] (\v-52) at (\v-52) {};
      \draw (\v-52) -- (p);
    }
    \node (ph) [style=Z dot] at ($(p) + (1.5*\dx,0)$) {};
    \node at (ph) [anchor=180] {$\scriptstyle \;\tfrac{\pi}{4}$};
    \draw (p) -- (ph);
    \end{tikzpicture}
  \\[-3.5ex]~
  \end{aligned}

%% file: diagrams/4linetgadgetdecom.tikz
\begin{tikzpicture}
	\begin{pgfonlayer}{nodelayer}
		\node  (0) at (-2.875, 0) {$=$};
		\node  (1) at (4.5, -1.25) {};
		\node [X dot] (3) at (-0.75, 1.25) {};
		\node [Z dot] (4) at (1.75, -0.5) {\footnotesize$\!\tau\!$};
		\node [Z dot] (5) at (-4.75, -0.25) {};
		\node  (6) at (-5.25, -1.25) {};
		\node [Z dot] (7) at (0.75, -0.25) {};
		\node  (8) at (-2.25, 0.75) {};
		\node [Z dot] (9) at (-3.75, -0.5) {};
		\node at (9) [anchor=west] {\footnotesize\;$\pi/4$};
		\node [Z dot] (10) at (-0.75, 1.75) {};
		\node [X dot] (11) at (-1.25, -0.75) {};
		\node  (12) at (-3.5, 0.5) {};
		\node [Z dot] (13) at (-4.75, 0.5) {};
		\node  (14) at (-5.25, 0.5) {};
		\node [Z dot] (15) at (-1.25, -0.25) {};
		\node  (16) at (-3.5, -1.25) {};
		\node [Z dot] (17) at (-1.25, -1.25) {};
		\node [Z dot] (18) at (0.75, -1.25) {};
		\node  (19) at (4.5, -0.25) {};
		\node  (20) at (-2.25, -0.25) {};
		\node  (21) at (-2.25, -1.25) {};
		\node [X dot] (22) at (-4.25, -0.5) {};
		\node [Z dot] (23) at (-1.75, -0.75) {\footnotesize$\!\iota\!$};
		\node  (24) at (4.5, 0.75) {};
		\node [Z dot] (25) at (-0.25, 1.25) {\footnotesize$\!\iota\!$};
		\node [X dot] (26) at (1.25, -0.5) {};
		\node  (27) at (-3.5, -0.25) {};
		\node  (28) at (-5.25, -0.25) {};
		\node [Z dot] (29) at (0.75, 0.75) {};
		\node [Z dot] (30) at (-0.75, -0.25) {};
		\node [Z dot] (31) at (-4.75, -1.25) {};
		\node  (32) at (-3.5, 1.25) {};
		\node [Z dot] (33) at (-4.75, 1.25) {};
		\node  (34) at (-5.25, 1.25) {};
		\node [Z dot] (35) at (-1.75, 1.75) {\footnotesize$\!\tau\!$};
		\node  (36) at (4.5, 1.75) {};
		\node  (37) at (-2.25, 1.75) {};
		\node [X dot] (38) at (-1.25, 1.25) {};
		\node [Z dot] (39) at (-1.25, 1.75) {};
		\node [Z dot] (40) at (-1.25, 0.75) {};
		\node [Z dot] (41) at (-1.75, 1.25) {\footnotesize$\!\iota\!$};
		\node [X dot] (42) at (-1.25, 0.25) {};
		\node [Z dot] (43) at (-1.25, 0.75) {};
		\node [Z dot] (44) at (-1.25, -0.25) {};
		\node [Z dot] (45) at (-1.75, 0.25) {\footnotesize$\!\iota\!$};
		\node [Z dot] (46) at (-1.75, 0.75) {\footnotesize$\!\tau\!$};
		\node [Z dot] (47) at (-1.75, -0.25) {\footnotesize$\!\tau\!$};
		\node [Z dot] (48) at (-1.75, -1.25) {\footnotesize$\!\tau\!$};
		\node [Z dot] (49) at (-0.5, -1.25) {};
		\node [Z dot] (50) at (-0.5, 0.75) {};
		\node [Z dot] (51) at (0, 0.25) {\footnotesize$\!\iota\!$};
		\node [X dot] (52) at (-0.5, 0.25) {};
		\node [Z dot] (53) at (0.25, -1.25) {};
		\node [Z dot] (54) at (0.25, 1.75) {};
		\node [Z dot] (55) at (0.75, 1.25) {\footnotesize$\!\iota\!$};
		\node [X dot] (56) at (0.25, 1.25) {};
		\node [Z dot] (57) at (1.5, 1.75) {};
		\node [X dot] (58) at (2.25, -0.75) {};
		\node [Z dot] (59) at (1.5, 0.75) {};
		\node [Z dot] (60) at (2.75, -0.75) {\footnotesize$\!\tau\!$};
		\node [Z dot] (61) at (1.75, -1.25) {};
		\node [Z dot] (62) at (2.5, 1.75) {};
		\node [X dot] (63) at (3, 0.5) {};
		\node [Z dot] (64) at (2.5, 0.75) {};
		\node [Z dot] (65) at (3.5, 0.5) {\footnotesize$\!\tau\!$};
		\node [Z dot] (66) at (2.5, -0.25) {};
		\node [Z dot] (67) at (3.5, 1.75) {};
		\node [X dot] (68) at (4, -0.5) {};
		\node [Z dot] (69) at (3.5, -0.25) {};
		\node [Z dot] (70) at (4.5, -0.5) {\footnotesize$\!\tau\!$};
		\node [Z dot] (71) at (3.5, -1.25) {};
	\end{pgfonlayer}
	\begin{pgfonlayer}{edgelayer}
		\draw (26) to (4);
		\draw (3) to (25);
		\draw (21.center) to (1.center);
		\draw (20.center) to (19.center);
		\draw (8.center) to (24.center);
		\draw (15) to (11);
		\draw (17) to (11);
		\draw (11) to (23);
		\draw (10) to (3);
		\draw (30) to (3);
		\draw (29) to (26);
		\draw (7) to (26);
		\draw (18) to (26);
		\draw (22) to (9);
		\draw (6.center) to (16.center);
		\draw (28.center) to (27.center);
		\draw (14.center) to (12.center);
		\draw (13) to (22);
		\draw (5) to (22);
		\draw (31) to (22);
		\draw (34.center) to (32.center);
		\draw (33) to (22);
		\draw (37.center) to (36.center);
		\draw (39) to (38);
		\draw (40) to (38);
		\draw (38) to (41);
		\draw (43) to (42);
		\draw (44) to (42);
		\draw (42) to (45);
		\draw (52) to (51);
		\draw (50) to (52);
		\draw (49) to (52);
		\draw (56) to (55);
		\draw (54) to (56);
		\draw (53) to (56);
		\draw (58) to (60);
		\draw (57) to (58);
		\draw (59) to (58);
		\draw (61) to (58);
		\draw (63) to (65);
		\draw (62) to (63);
		\draw (64) to (63);
		\draw (66) to (63);
		\draw (68) to (70);
		\draw (67) to (68);
		\draw (69) to (68);
		\draw (71) to (68);
	\end{pgfonlayer}
\end{tikzpicture}

%% file: diagrams/4linetgadgetdecominduc.tikz
\begin{tikzpicture}
	\begin{pgfonlayer}{nodelayer}
		\node  (0) at (-3, 1.25) {$=$};
		\node  (1) at (-5.5, 2.25) {};
		\node  (2) at (-5.5, 0.25) {};
		\node [Z dot] (3) at (-1.25, 0.25) {};
		\node  (4) at (-2.25, 0.25) {};
		\node  (5) at (-1.5, 1) {$\vdots$};
		\node  (7) at (-2.5, 1.75) {};
		\node [X dot] (8) at (-1, 1) {};
		\node  (9) at (-2.25, 1.75) {};
		\node  (10) at (-0.25, 0.25) {};
		\node [Z dot] (11) at (-0.5, 1) {};
		\node  (12) at (-0.25, 1.75) {};
		\node [anchor=south] (6) at (-0.25,0.375) {\footnotesize$\pi\!\!\;/\!\!\:4$};
		\node  (13) at (-2.5, 0.25) {};
		\node [Z dot] (14) at (-1.25, 1.75) {};
		\node  (15) at (-2.25, 2.5) {};
		\node  (16) at (-0.25, 2.5) {};
		\node [Z dot] (17) at (-0.75, 2.5) {};
		\node [X dot] (18) at (-0.75, 1.75) {};
		\node [Z dot] (19) at (-1.75, 2.5) {};
		\node [X dot] (20) at (-1.75, 1.75) {};
		\node  (21) at (0.25, 1.25) {$=$};
		\node [Z dot] (22) at (-3.75, 1) {};
		\node  (23) at (-3.5, 1.25) {};
		\node  (24) at (-5.25, 1.25) {};
		\node  (25) at (-3.5, 0.75) {\footnotesize$\pi\!\!\;/\!\!\:4$};
		\node  (26) at (-4.75, 0.75) {$\vdots$};
		\node [Z dot] (27) at (-4.75, 0.25) {};
		\node [Z dot] (28) at (-4.75, 1.25) {};
		\node  (29) at (-4.75, 1.75) {$\vdots$};
		\node  (30) at (-5.25, 0.25) {};
		\node  (31) at (-3.5, 2.25) {};
		\node [X dot] (32) at (-4.25, 1) {};
		\node [Z dot] (33) at (-4.75, 2.25) {};
		\node  (34) at (-5.25, 2.25) {};
		\node  (35) at (-3.5, 0.25) {};
		\node  (36) at (7, 0.25) {};
		\node [Z dot] (37) at (6.5, 1) {};
		\node [Z dot] (38) at (5.5, 1.25) {};
		\node  (39) at (1.25, 0.75) {$\vdots$};
		\node  (40) at (6.75, 0.75) {\footnotesize$\pi\!\!\;/\!\!\:4$};
		\node  (41) at (1.5, 2.25) {};
		\node [X dot] (42) at (3.5, 0.75) {};
		\node [Z dot] (43) at (3.5, 1.25) {};
		\node  (44) at (1.25, 1.75) {$\vdots$};
		\node [Z dot] (45) at (3.5, 0.25) {};
		\node [Z dot] (46) at (5.5, 0.25) {};
		\node  (47) at (7, 1.25) {};
		\node  (48) at (1.5, 1.25) {};
		\node [Z dot] (49) at (2.5, 2.25) {\footnotesize$\!\!\sigma_m\!\!$};
		\node  (50) at (1.5, 0.25) {};
		\node [Z dot] (51) at (4, 0.75) {\footnotesize$\!\!\theta_m\!\!$};
		\node  (52) at (7, 2.25) {};
		\node [X dot] (53) at (6, 1) {};
		\node [Z dot] (54) at (5.5, 2.25) {};
		\node  (55) at (7, 3) {};
		\node  (56) at (1.5, 3) {};
		\node [Z dot] (57) at (2, 3) {};
		\node [X dot] (58) at (2, 2.25) {};
		\node [Z dot] (59) at (6, 3) {};
		\node [X dot] (60) at (6, 2.25) {};
		\node [Z dot] (61) at (2.5, 0.25) {\footnotesize$\!\!\sigma_m\!\!$};
		\node [Z dot] (62) at (2.5, 1.25) {\footnotesize$\!\!\sigma_m\!\!$};
		\node [Z dot] (63) at (3, 1.25) {};
		\node [Z dot] (64) at (3.5, 1.75) {\footnotesize$\!\!\theta_m\!\!$};
		\node [Z dot] (65) at (3, 2.25) {};
		\node [X dot] (66) at (3, 1.75) {};
		\node [Z dot] (67) at (4.5, 0.25) {};
		\node [Z dot] (68) at (5, 1.75) {\footnotesize$\!\!\theta_m\!\!$};
		\node [Z dot] (69) at (4.5, 2.25) {};
		\node [X dot] (70) at (4.5, 1.75) {};
		\node  (71) at (1, 0.25) {$m+1$};
		\node  (72) at (1.25, 1.25) {$i$};
		\node  (73) at (1.25, 3) {$1$};
		\node  (74) at (1.25, 2.25) {$2$};
		\node  (75) at (1.75, -0.75) {};
		\node [X dot] (76) at (-2, -3.25) {};
		\node [X dot] (77) at (-2.25, -2.25) {};
		\node [Z dot] (78) at (-0.25, -1.75) {};
		\node [Z dot] (79) at (-1.5, -3.25) {\footnotesize$\!\!\theta_m\!\!$};
		\node [Z dot] (80) at (1, -3) {};
		\node [X dot] (81) at (-1, -2.25) {};
		\node  (82) at (-4.75, -1.75) {$2$};
		\node [Z dot] (83) at (-2.5, -2.75) {};
		\node  (84) at (-4.75, -3.25) {$\vdots$};
		\node  (85) at (-4.5, -0.75) {};
		\node [Z dot] (86) at (-1.75, -2.25) {\footnotesize$\!\!\theta_m\!\!$};
		\node [Z dot] (87) at (-2.5, -1.75) {};
		\node  (88) at (-4.5, -3.75) {};
		\node  (89) at (1.75, -2.75) {};
		\node [anchor=south] at (1.25,-3.5) {\footnotesize$\pi\!\!\;/\!\!\:4$};
		\node  (91) at (-4.75, -2.75) {$i$};
		\node  (92) at (-4.75, -0.75) {$1$};
		\node  (93) at (-4.5, -2.75) {};
		\node  (94) at (-5.5, -2.75) {$=$};
		\node [Z dot] (95) at (-2, -3.75) {};
		\node [Z dot] (96) at (-1.25, -3.75) {};
		\node  (97) at (-5, -3.75) {$m+1$};
		\node [Z dot] (98) at (-2, -2.75) {};
		\node [Z dot] (99) at (-1.5, -1.75) {};
		\node  (100) at (-4.5, -1.75) {};
		\node [Z dot] (101) at (-0.5, -2.25) {\footnotesize$\!\!\theta_m\!\!$};
		\node [Z dot] (102) at (-3.5, -3.75) {\footnotesize$\!\!\sigma_m\!\!$};
		\node [Z dot] (103) at (-0.25, -2.75) {};
		\node [Z dot] (104) at (-0.25, -3.75) {};
		\node  (105) at (1.75, -3.75) {};
		\node  (106) at (-4.75, -2.25) {$\vdots$};
		\node [X dot] (107) at (0.5, -3) {};
		\node [Z dot] (108) at (-3.5, -2.75) {\footnotesize$\!\!\sigma_m\!\!$};
		\node  (109) at (1.75, -1.75) {};
		\node [X dot] (110) at (-3.5, -1.25) {};
		\node [Z dot] (111) at (-3.5, -1.75) {};
		\node [Z dot] (112) at (-3, -1.25) {\footnotesize$\!\!\sigma_m\!\!$};
		\node [Z dot] (113) at (-3.5, -0.75) {};
		\node [Z dot] (114) at (-2.5, -0.75) {};
		\node [Z dot] (115) at (-1.5, -0.75) {};
		\node [Z dot] (116) at (-0.25, -0.75) {};
	\end{pgfonlayer}
	\begin{pgfonlayer}{edgelayer}
		\draw [style=braceedge] (2.center) to node[wire label, inner sep=5 pt]{$m+1$} (1.center);
		\draw (8) to (11);
		\draw (4.center) to (10.center);
		\draw (9.center) to (12.center);
		\draw (14) to (8);
		\draw (3) to (8);
		\draw [style=braceedge] (13.center) to node[wire label, inner sep=5 pt]{$m$} (7.center);
		\draw (15.center) to (16.center);
		\draw (17) to (18);
		\draw (19) to (20);
		\draw (32) to (22);
		\draw (30.center) to (35.center);
		\draw (24.center) to (23.center);
		\draw (34.center) to (31.center);
		\draw (33) to (32);
		\draw (28) to (32);
		\draw (27) to (32);
		\draw (53) to (37);
		\draw (50.center) to (36.center);
		\draw (48.center) to (47.center);
		\draw (41.center) to (52.center);
		\draw (43) to (42);
		\draw (45) to (42);
		\draw (42) to (51);
		\draw (54) to (53);
		\draw (38) to (53);
		\draw (46) to (53);
		\draw (56.center) to (55.center);
		\draw (57) to (58);
		\draw (59) to (60);
		\draw (65) to (66);
		\draw (63) to (66);
		\draw (66) to (64);
		\draw (69) to (70);
		\draw (67) to (70);
		\draw (70) to (68);
		\draw (107) to (80);
		\draw (88.center) to (105.center);
		\draw (93.center) to (89.center);
		\draw (100.center) to (109.center);
		\draw (98) to (76);
		\draw (95) to (76);
		\draw (76) to (79);
		\draw (78) to (107);
		\draw (103) to (107);
		\draw (104) to (107);
		\draw (85.center) to (75.center);
		\draw (87) to (77);
		\draw (83) to (77);
		\draw (77) to (86);
		\draw (99) to (81);
		\draw (96) to (81);
		\draw (81) to (101);
		\draw (113) to (110);
		\draw (111) to (110);
		\draw (110) to (112);
		\draw [bend left=15, looseness=1.00] (114) to (77);
		\draw (115) to (81);
		\draw (116) to (107);
	\end{pgfonlayer}
\end{tikzpicture}

%% file: parity-phase-reduction-ZX-TODD.bbl
\begin{thebibliography}{10}
\providecommand{\bibitemdeclare}[2]{}
\providecommand{\surnamestart}{}
\providecommand{\surnameend}{}
\providecommand{\urlprefix}{Available at }
\providecommand{\url}[1]{\texttt{#1}}
\providecommand{\href}[2]{\texttt{#2}}
\providecommand{\urlalt}[2]{\href{#1}{#2}}
\providecommand{\doi}[1]{doi:\urlalt{http://dx.doi.org/#1}{#1}}
\providecommand{\bibinfo}[2]{#2}

\bibitemdeclare{inproceedings}{Amy-2018}
\bibitem{Amy-2018}
\bibinfo{author}{Matthew \surnamestart Amy\surnameend} (\bibinfo{year}{2018}):
  \emph{\bibinfo{title}{Towards Large-scale Functional Verification of
  Universal Quantum Circuits}}.
\newblock In: {\sl \bibinfo{booktitle}{Proceedings of QPL 2018}}, pp.
  \bibinfo{pages}{1--21}, \doi{10.4204/EPTCS.287.1}.
\newblock \bibinfo{note}{{[arXiv:1901.09476]; see also
  [\url{https://github.com/meamy/feynman}].}}

\bibitemdeclare{article}{ACR-2018}
\bibitem{ACR-2018}
\bibinfo{author}{Matthew \surnamestart Amy\surnameend},
  \bibinfo{author}{Jianxin \surnamestart Chen\surnameend} \&
  \bibinfo{author}{Neil~J. \surnamestart Ross\surnameend}
  (\bibinfo{year}{2018}): \emph{\bibinfo{title}{A Finite Presentation of
  CNOT-Dihedral Operators}}.
\newblock {\sl \bibinfo{journal}{Electronic Proceedings in Theoretical Computer
  Science}} \bibinfo{volume}{266}, pp. \bibinfo{pages}{84--97},
  \doi{10.1007/978-3-642-12821-9_4}.
\newblock \bibinfo{note}{{[arXiv:1701.00140]}}.

\bibitemdeclare{article}{AMM-2014}
\bibitem{AMM-2014}
\bibinfo{author}{Matthew \surnamestart {Amy}\surnameend},
  \bibinfo{author}{Dmitri \surnamestart {Maslov}\surnameend} \&
  \bibinfo{author}{Michele \surnamestart {Mosca}\surnameend}
  (\bibinfo{year}{2014}): \emph{\bibinfo{title}{Polynomial-Time {T}-Depth
  Optimization of {Clifford}+{T} Circuits Via Matroid Partitioning}}.
\newblock {\sl \bibinfo{journal}{{IEEE} {Transactions} on {Computer}-{Aided}
  {Design} of {Integrated} {Circuits} and {Systems}}}
  \bibinfo{volume}{33}(\bibinfo{number}{10}), pp. \bibinfo{pages}{1476--1489},
  \doi{10.1109/TCAD.2014.2341953}.
\newblock \bibinfo{note}{{[arXiv:1303.2042]}}.

\bibitemdeclare{article}{AMMR-2013}
\bibitem{AMMR-2013}
\bibinfo{author}{Matthew \surnamestart Amy\surnameend}, \bibinfo{author}{Dmitri
  \surnamestart Maslov\surnameend}, \bibinfo{author}{Michele \surnamestart
  Mosca\surnameend} \& \bibinfo{author}{Martin \surnamestart
  Roetteler\surnameend} (\bibinfo{year}{2013}): \emph{\bibinfo{title}{A
  meet-in-the-middle algorithm for fast synthesis of depth-optimal quantum
  circuits}}.
\newblock {\sl \bibinfo{journal}{IEEE Transactions on Computer-Aided Design of
  Integrated Circuits and Systems}} \bibinfo{volume}{32}(\bibinfo{number}{6}),
  pp. \bibinfo{pages}{818--830}, \doi{10.1109/TCAD.2013.2244643}.
\newblock \bibinfo{note}{{[arXiv:1206.0758]}}.

\bibitemdeclare{article}{AM-2019}
\bibitem{AM-2019}
\bibinfo{author}{Matthew \surnamestart {Amy}\surnameend} \&
  \bibinfo{author}{Michele \surnamestart {Mosca}\surnameend}
  (\bibinfo{year}{2019}): \emph{\bibinfo{title}{{T}-count optimization and
  {Reed}-{Muller} codes}}.
\newblock {\sl \bibinfo{journal}{IEEE Transactions on Information Theory}}
  \bibinfo{volume}{65}(\bibinfo{number}{8}), pp. \bibinfo{pages}{4771--4784},
  \doi{10.1109/TIT.2019.2906374}.
\newblock \bibinfo{note}{{[arXiv:1601.07363]}}.

\bibitemdeclare{article}{CH-2017}
\bibitem{CH-2017}
\bibinfo{author}{Earl~T. \surnamestart Campbell\surnameend} \&
  \bibinfo{author}{Mark \surnamestart Howard\surnameend}
  (\bibinfo{year}{2017}): \emph{\bibinfo{title}{A unified framework for magic
  state distillation and multi-qubit gate-synthesis with reduced resource
  cost}}.
\newblock {\sl \bibinfo{journal}{Physical Review A}} \bibinfo{volume}{95}, p.
  \bibinfo{pages}{022316}, \doi{10.1103/PhysRevA.86.022316}.
\newblock \bibinfo{note}{{[arXiv:1606.01904]}}.

\bibitemdeclare{inproceedings}{DP-2010}
\bibitem{DP-2010}
\bibinfo{author}{Ross \surnamestart Duncan\surnameend} \&
  \bibinfo{author}{Simon \surnamestart Perdrix\surnameend}
  (\bibinfo{year}{2010}): \emph{\bibinfo{title}{Rewriting Measurement-Based
  Quantum Computations with Generalised Flow}}.
\newblock In \bibinfo{editor}{Samson \surnamestart Abramsky\surnameend},
  \bibinfo{editor}{Cyril \surnamestart Gavoille\surnameend},
  \bibinfo{editor}{Claude \surnamestart Kirchner\surnameend},
  \bibinfo{editor}{Friedhelm \surnamestart Meyer auf~der Heide\surnameend} \&
  \bibinfo{editor}{Paul~G. \surnamestart Spirakis\surnameend}, editors: {\sl
  \bibinfo{booktitle}{Automata, Languages and Programming}},
  \bibinfo{publisher}{Springer Berlin Heidelberg}, \bibinfo{address}{Berlin,
  Heidelberg}, pp. \bibinfo{pages}{285--296}, \doi{10.1007/s10472-009-9141-x}.

\bibitemdeclare{article}{Gidney-2018}
\bibitem{Gidney-2018}
\bibinfo{author}{Craig \surnamestart Gidney\surnameend} (\bibinfo{year}{2018}):
  \emph{\bibinfo{title}{Halving the cost of quantum addition}}.
\newblock {\sl \bibinfo{journal}{Quantum}} \bibinfo{volume}{2},
  p.~\bibinfo{pages}{74}, \doi{10.1007/s11128-011-0297-z}.
\newblock \bibinfo{note}{{[arXiv:1709.06648]}}.

\bibitemdeclare{article}{GKMR-2014}
\bibitem{GKMR-2014}
\bibinfo{author}{David \surnamestart Gosset\surnameend}, \bibinfo{author}{Vadym
  \surnamestart Kliuchnikov\surnameend}, \bibinfo{author}{Michele \surnamestart
  Mosca\surnameend} \& \bibinfo{author}{Vincent \surnamestart Russo\surnameend}
  (\bibinfo{year}{2014}): \emph{\bibinfo{title}{An Algorithm for the T-count}}.
\newblock {\sl \bibinfo{journal}{Quantum Info. Comput.}}
  \bibinfo{volume}{14}(\bibinfo{number}{15-16}), pp.
  \bibinfo{pages}{1261--1276}.
\newblock \urlprefix\url{http://dl.acm.org/citation.cfm?id=2685179.2685180}.
\newblock \bibinfo{note}{{[arXiv:1308.4134]}}.

\bibitemdeclare{article}{HC-2018}
\bibitem{HC-2018}
\bibinfo{author}{Luke~E. \surnamestart Heyfron\surnameend} \&
  \bibinfo{author}{Earl~T. \surnamestart Campbell\surnameend}
  (\bibinfo{year}{2018}): \emph{\bibinfo{title}{An efficient quantum compiler
  that reduces T count}}.
\newblock {\sl \bibinfo{journal}{Quantum Science and Technology}}
  \bibinfo{volume}{4}(\bibinfo{number}{1}), p. \bibinfo{pages}{015004},
  \doi{10.1038/srep01939}.
\newblock \bibinfo{note}{{[arXiv:1712.01557]}}.

\bibitemdeclare{article}{Jones-2013}
\bibitem{Jones-2013}
\bibinfo{author}{Cody \surnamestart Jones\surnameend} (\bibinfo{year}{2013}):
  \emph{\bibinfo{title}{Low-overhead constructions for the fault-tolerant
  Toffoli gate}}.
\newblock {\sl \bibinfo{journal}{Phys. Rev. A}} \bibinfo{volume}{87}, p.
  \bibinfo{pages}{022328}, \doi{10.1103/PhysRevA.87.022328}.
\newblock \urlprefix\url{https://link.aps.org/doi/10.1103/PhysRevA.87.022328}.
\newblock \bibinfo{note}{{[arXiv:1212.5069]}}.

\bibitemdeclare{unpublished}{KvdW-2019}
\bibitem{KvdW-2019}
\bibinfo{author}{Aleks \surnamestart Kissinger\surnameend} \&
  \bibinfo{author}{John \surnamestart van~de Wetering\surnameend}
  (\bibinfo{year}{2019}): \emph{\bibinfo{title}{Reducing T-count with the
  ZX-calculus}}.
\newblock \bibinfo{note}{{[arXiv:1903.10477]}}.

\bibitemdeclare{article}{Litinski-2019}
\bibitem{Litinski-2019}
\bibinfo{author}{Daniel \surnamestart Litinski\surnameend}
  (\bibinfo{year}{2019}): \emph{\bibinfo{title}{{A} {G}ame of {S}urface
  {C}odes: Large-Scale Quantum Computing with Lattice Surgery}}.
\newblock {\sl \bibinfo{journal}{Quantum}} \bibinfo{volume}{3}, p.
  \bibinfo{pages}{128}, \doi{10.1103/PhysRevB.96.205413}.
\newblock \bibinfo{note}{{[arXiv:1808.02892]}}.

\bibitemdeclare{article}{MR-2018}
\bibitem{MR-2018}
\bibinfo{author}{Dmitri \surnamestart Maslov\surnameend} \&
  \bibinfo{author}{Martin \surnamestart Roetteler\surnameend}
  (\bibinfo{year}{2018}): \emph{\bibinfo{title}{Shorter stabilizer circuits via
  {B}ruhat decomposition and quantum circuit transformations}}.
\newblock {\sl \bibinfo{journal}{{IEEE} Transactions on Information Theory}}
  \bibinfo{volume}{64}, pp. \bibinfo{pages}{4729--4738},
  \doi{10.1109/TIT.2018.2825602}.
\newblock \bibinfo{note}{{[arXiv:1705.09176]}}.

\bibitemdeclare{unpublished}{MSCRdM-2019}
\bibitem{MSCRdM-2019}
\bibinfo{author}{Giulia \surnamestart Meuli\surnameend},
  \bibinfo{author}{Mathias \surnamestart Soeken\surnameend},
  \bibinfo{author}{Earl \surnamestart Campbell\surnameend},
  \bibinfo{author}{Martin \surnamestart Roetteler\surnameend} \&
  \bibinfo{author}{Giovanni~De \surnamestart Micheli\surnameend}
  (\bibinfo{year}{2019}): \emph{\bibinfo{title}{The Role of Multiplicative
  Complexity in Compiling Low {T}-count Oracle Circuits}}.
\newblock \bibinfo{note}{{[arXiv:1908.01609]}}.

\bibitemdeclare{misc}{Quipper}
\bibitem{Quipper}
\bibinfo{author}{Peter \surnamestart Selinger\surnameend}:
  \emph{\bibinfo{title}{Quipper}}.
\newblock
  \bibinfo{howpublished}{\url{https://www.mathstat.dal.ca/~selinger/quipper}}.

\bibitemdeclare{unpublished}{ZhangCheng19}
\bibitem{ZhangCheng19}
\bibinfo{author}{Fang \surnamestart {Zhang}\surnameend} \&
  \bibinfo{author}{Jianxin \surnamestart {Chen}\surnameend}
  (\bibinfo{year}{2019}): \emph{\bibinfo{title}{Optimizing T gates in
  Clifford+T circuit as $\pi/4$ rotations around Paulis}}.
\newblock \bibinfo{note}{{[arXiv:1903.12456]}}.

\end{thebibliography}
